\documentclass[3p,authoryear,preprint,12pt]{elsarticle}

\usepackage{amssymb}
\usepackage{algorithm}
\usepackage{algorithmic}
\usepackage{amsmath}\allowdisplaybreaks

\newcommand{\ap}{\mathrm{ap}} % activation probability
\newcommand{\E}{\mathbb{E}} % expectation
\newcommand{\OPT}{{\it S^*}}
\newcommand{\ALG}{{\it S^A}}
\newcommand{\cond}{\,|\,}  % this vertical line seems to have better spacing than \mid, or simply |

\newtheorem{theorem}{Theorem}
\newtheorem{assumption}{Assumption}
\newtheorem{claim}{Claim}
\newtheorem{corollary}{Corollary}
\newtheorem{lemma}{Lemma}
\newproof{proof}{Proof}

\DeclareMathOperator*{\argmax}{argmax}

\journal{}

\begin{document}

\begin{frontmatter}



\title{Network Inference and Influence Maximization from Samples\tnoteref{t1,t2}}
\tnotetext[t1]{A preliminary version containing only results about the IC model appears in ICML 2021.}
\tnotetext[t2]{This work was supported in part by the National Natural Science Foundation of China Grants No.~61832003, 61872334, the Strategic Priority Research Program of Chinese Academy of Sciences under Grant No.~XDA27000000, the 973 Program of China Grant No.~2016YFB1000201, K.C.~Wong Education Foundation.}


\author[inst1,inst2]{Zhijie Zhang\corref{cor1}}
\ead{zhangzhijie@ict.ac.cn}
\author[inst3]{Wei Chen}
\ead{weic@microsoft.com}
\author[inst1,inst2]{Xiaoming Sun}
\ead{sunxiaoming@ict.ac.cn}
\author[inst1,inst2]{Jialin Zhang}
\ead{zhangjialin@ict.ac.cn}

\cortext[cor1]{Corresponding author}

\affiliation[inst1]{organization={Institute of Computing Technology, Chinese Academy of Sciences},%Department and Organization
            % addressline={Address One}, 
            city={Beijing},
            postcode={100080}, 
            % state={Beijing},
            country={China}}
            
\affiliation[inst2]{organization={University of Chinese Academy of Sciences},%Department and Organization
            % addressline={Address One}, 
            city={Beijing},
            postcode={100080}, 
            % state={Beijing},
            country={China}}

\affiliation[inst3]{organization={Microsoft Research Asia},%Department and Organization
            % addressline={Address Two}, 
            city={Beijing},
            postcode={100080}, 
            % state={State Two},
            country={China}}

\begin{abstract}
Influence maximization is the task of selecting a small number of seed nodes in a social network to maximize the influence spread from these seeds. It has been widely investigated in the past two decades. In the canonical setting, the social network and its diffusion parameters are given as input. In this paper, we consider the more realistic sampling setting where the network is unknown and we only have a set of passively observed cascades that record the sets of activated nodes at each diffusion step. We study the task of influence maximization from these cascade samples (IMS) and present constant approximation algorithms for it under mild conditions on the seed set distribution. To achieve the optimization goal, we also provide a novel solution to the network inference problem, that is, learning diffusion parameters and the network structure from the cascade data. Compared with prior solutions, our network inference algorithms require weaker assumptions and do not rely on maximum-likelihood estimation and convex programming. Our IMS algorithms enhance the learning-and-then-optimization approach by allowing a constant approximation ratio even when the diffusion parameters are hard to learn, and we do not need any assumption related to the network structure or diffusion parameters.
\end{abstract}

%\includegraphics[scale=0.6]{OPSS-graphical abstract.pdf}
%\item Brand new network inference algorithms under both IC and LT models to learn edge parameters from cascade samples.
%Our algorithms enjoy faster implementation, lower sample complexity and weaker assumptions comparing to previous algorithms.
%Our assumptions are also easy to verify from cascade samples.
%\item Several end-to-end algorithms for influence maximization from cascade samples (IMS) with constant approximation guarantee under both IC and LT models. Our algorithms work for any networks, even when it is impossible to learn edge parameters. These are the first such IMS algorithms in the literature.

\begin{keyword}
influence maximization \sep network inference \sep optimization from samples \sep data-driven optimization \sep end-to-end optimization
%%% MSC codes here, in the form: \MSC code \sep code
%%% or \MSC[2008] code \sep code (2000 is the default)
\end{keyword}

\end{frontmatter}


\section{Introduction}

Maximizing the spread of influence through a social network has been widely studied in the past two decades.
It models the phenomenon in which a small set of initially \emph{active} nodes called \emph{seeds} takes some piece of information (news, ideas or opinions, etc.), and the information spreads over the network to \emph{activate} the remaining nodes.
The expected number of \emph{final} active nodes is called the \emph{influence spread} of the seed set.
The \emph{influence maximization} problem asks to pick at most $k$ seeds in order to maximize the influence spread.
Under many \emph{diffusion models} such as the discrete-time \emph{independent cascade} (IC) model and \emph{linear threshold} (LT) model \citep{KempeKT03}, the problem enjoys a $(1-1/e-\varepsilon)$-approximation (with small $\varepsilon> 0$), which is tight assuming P $\neq$ NP \citep{Feige98}.
It has found applications in many scenarios.

Traditional influence maximization problem requires as input the whole social network (as well as its parameters), based on which one can compute or estimate the influence spread function.
In many scenarios, however, this might be too demanding, especially for those who do not have free access to the network.
In this work, we consider influence maximization in the sampling setting where one only has access to a set of passively observed cascades spreading over an implicit social network.
Each cascade records the sets of activated nodes at each time step. 
Such sample data is available in many scenarios, especially on the Internet where the timestamps can be recorded in principle.
We are interested in whether we can maximize the influence from such sample data.
We model this problem as \emph{influence maximization from samples} below:

\begin{quotation}
	\noindent\textbf{Influence maximization from samples (IMS).}
	For an unknown social network $G$ with diffusion parameters,  given 
	$t$ cascade samples where each seed is independently sampled with an unknown probability,  
	can we find a seed set of size at most $k$ such that its influence spread is a constant approximation of the optimal seed set, when $t$ is polynomial to the size of $G$?
\end{quotation}

En route to solving the above problem, a natural and reasonable approach is to first learn the network structure as well as its parameters, and then maximize the influence over the learned network.
This leads to the well-studied \emph{network inference} problem below:

\begin{quotation}
	\noindent\textbf{Network inference.} For an unknown social network $G$, 
	given polynomial number of cascade samples where each seed is sampled independently with an unknown probability, 
	estimate all diffusion parameters such that with probability at least $1-\delta$, every parameter is estimated within an additive error $\varepsilon$.
\end{quotation}

%\noindent\textbf{Network Inference.} \emph{Given an unknown social network $G=(V,E,p)$ and samples $(S_{i,0},S_{i,1},\cdots,S_{i,n-1})_{i=1}^t$ as above, estimate the values of $p_{uv}$ within an additive error. More formally, for some $\varepsilon\in(0,1)$, compute a vector $\hat{p}$ such that $|\hat{p}_{uv}-p_{uv}|\leq\varepsilon$ for all $u,v\in V$.\footnote{In the literature, network inference usually refers to recovering the edge set of the underlying graph.
%		Here we slightly abuse the notation since the task is often fulfilled by estimating the network parameters.}}
%
%\wei{
%Here is my rewrite to avoid undefined notations.
%
%\noindent\textbf{Network Inference.} \emph{For an unknown social network $G$, 
%	given polynomial number of cascade samples where each seed is sampled independently with an unknown probability, 
%	estimate all IC model parameters such that with probability at least $1-\delta$, every parameter is estimated within a small additive error $\varepsilon$.
%}
%}

Our contributions in this work are mainly two-fold.
First, we revisit the network inference problem and design  brand new algorithms for it under both IC and LT models.
For the IC model, while all previous algorithms are based on the maximum likelihood estimation and convex programming 
\citep{NetrapalliS12,NarasimhanPS15,Pouget-AbadieH15}, 
our algorithm builds on a closed-form expression for each edge probability in terms of quantities which can be efficiently estimated.
As a result, our algorithm enjoys faster implementation, lower sample complexity and weaker assumptions comparing to previous algorithms.
Our assumptions are also easy to verify from cascade samples.
We will discuss these differences further in the end of Section \ref{section: edge probabilities estimation}.
For the LT model, to the best of our knowledge, there is no network inference algorithm with theoretical guarantees previously.
To resolve the problem, we also build a closed-form expression for each edge weight, which makes the estimation possible.

Second, we provide several end-to-end IMS algorithms with constant approximation guarantee under both IC and LT models.
For the IC model, following the canonical learning-and-then-optimization framework, we first present an IMS algorithm by directly invoking our network inference algorithm.
The algorithm thus needs the assumptions used for learning.
Next, we present alternative algorithms which only need two simple assumptions on the seed distribution and impose no requirements for the underlying network.
In contrast, all the known algorithms for network inference (including ours) impose some restrictions on the network.
This result is highly non-trivial since it is impossible to resolve network inference problem on arbitrary graphs and hence the learning-and-then-optimization framework fails in this case.
For instance, consider a complete graph and another graph with one edge removed from the complete graph,
where all edge probabilities are $1$.
If each node is picked as a seed independently with probability $1/2$, one cannot distinguish them within polynomially many cascade samples.
For the LT model, an interesting feature of the network inference algorithm is that the learning assumption already has no requirements for the network.
Thus, the learning-and-then-optimization framework directly leads to an IMS algorithm that works for arbitrary networks under the LT model.
Our IMS follows the general optimization-from-samples framework \citep{BalkanskiRS17}, and
generalizes the recent result on optimization from structured samples for coverage functions \citep{ChenSZZ20}, see Section \ref{section: related work} for details.
Finally, we remark that our results not only apply to influence maximization, but also to other learning and optimization settings such as
probabilistic maximum cover with application in online advertising \citep{CWYW16}.

\subsection{Related Work}
\label{section: related work}

Influence maximization from samples follows the framework of \emph{optimization from samples} (OPS) originally proposed by \citet{BalkanskiRS17}: 
given a set of polynomial number of samples $\{S_i,f(S_i)\}_{i=1}^t$ and constraint $\mathcal{M}$, can we find a set $S\in\mathcal{M}$ such that $f(S)\geq c\cdot \max_{T\in\mathcal{M}} f(T)$ for some constant $c$?
The OPS framework is very important for the data-driven integration of learning and optimization where the underlying model (function $f$ above) is not
readily known.
Surprisingly, \citet{BalkanskiRS17} showed that even for the maximum coverage problem, there is no constant approximation algorithm under the OPS model,
despite prior results that a coverage function $f$ is learnable from samples \citep{BadanidiyuruDFKNR12} and constant optimization is available when $f$ is known \citep{NemhauserWF78}.
Subsequently, several attempts \citep{BalkanskiRS16,BalkanskiIS17,RosenfeldBGS18,ChenSZZ20} have been made to circumvent the impossibility result of \citet{BalkanskiRS17}.
Among them the most related one is the \emph{optimization from structured samples} (OPSS) model for coverage functions \citep{ChenSZZ20}, where the samples carry additional structural information in the form of $\{S_i,N(S_i)\}_{i=1}^t$, where $N(S_i)$ contains the nodes covered by $S_i$.
It was shown that if the samples are generated from a \emph{``negatively correlated''} distribution, the maximum coverage problem enjoys constant approximation in the OPSS model.
Recall that coverage functions can be regarded as IC influence spread functions defined over a bipartite graph with edge probabilities in $\{0,1\}$.
Thus, our result on IMS greatly generalizes OPSS to allow general graphs and stochastic cascades over edges.

End-to-end influence maximization from data has been explored by \citet{GoyalBL11}, but they only used a heuristic method to learn influence spread functions and then used the greedy
method for influence maximization, so there was no end-to-end guarantee on IMS.
A recent work \citep{BalkanskiIS17} revisited IMS problem under the OPS model, and provided a constant approximation algorithm when the underlying network is generated from the \emph{stochastic block} model.
Our study is the first to provide IMS algorithms with theoretical guarantees that work on arbitrary networks.

Network inference has been extensively studied over the past decade \citep{Gomez-RodriguezLK10,MyersL10,Gomez-RodriguezBS11,DuSSY12,NetrapalliS12,AbrahaoCKP13,DaneshmandGSS14,DuSGZ13,DuLBS14,NarasimhanPS15,Pouget-AbadieH15,HeX0L16}.
While most of them focused on the continuous time diffusion model, there are several results under the discrete time IC model \citep{NetrapalliS12,NarasimhanPS15,Pouget-AbadieH15}, all of which build on the maximum likelihood estimation.
We will compare these results with ours after we present our approach in Section \ref{section: edge probabilities estimation}.
%Netrapalli and Sanghavi \yrcite{NetrapalliS12} first considered network inference under the discrete time IC model.
%The key ingredient in their approach is an appropriate change of parameters such that the likelihood function is convex, enabling both efficient implementation and analysis.
%Their algorithm can be implemented in a local or distributed way, i.e.~the in-neighbors of each node can be learned parallelly.
%However, their algorithm works under a strong assumption called \emph{correlation decay}, which makes the result restricted.
%The approach was further developed in \cite{NarasimhanPS15} to solve edge probabilities estimation.
%However, unlike in \cite{NetrapalliS12}, the algorithm in \cite{NarasimhanPS15} requires the graph structure to be known in advance.
%Almost simultaneously, Pouget-Abadie and Horel \cite{Pouget-AbadieH15} presented an algorithm under \emph{regularity assumptions} and the \emph{restricted eigenvalue condition} \cite{BickelRT09}.
%Their algorithm estimates the edge parameters to an additive error under the so-called \emph{generalized linear cascade model}, which includes the IC model and the voter model.
%Their algorithm is optimal in the sample complexity.

%\cite{AminHK14}

\subsection{Organization}
In Section \ref{section: preliminaries}, we describe the model, some concepts and notations as well as two Chernoff-type lemmas used in the analysis.
In Section \ref{section: the independent cascade model}, we present network inference and IMS algorithms under the IC model.
In Section \ref{section: the linear threshold model}, we present network inference and IMS algorithms under the LT model.
Finally, we conclude the paper in Section \ref{section: conclusion}.

\section{Preliminaries}
\label{section: preliminaries}

%\zhijie{In LT model, there is an unknown social network $G=(V,E,w)$.
%	Each node $v\in V$ is associated with a threshold $\theta_v$, which is drawn independently and uniformly from $[0,1]$.
%	Let $(S_0,S_1,\cdots,S_{n-1})$ be the propagation process, where $S_0$ is drawn from some distribution $\mathcal{D}$, known as the seed distribution.
%	For node $v\in V$, let $N(v)=N^{in}(v)$ be the in-neighbors of $v$.
%	Let $\mathcal{D}_v$ be a marginal distribution of $\mathcal{D}$ restricted on $2^{N(v)}$.
%	For $u\in N(v)$, let $q_u=\Pr_{\mathcal{D}_v}[u\in S_0]$ be the probability that $u$ is picked as a seed.
%	Finally, let $\ap(v)=\Pr_{\mathcal{D}_v,\theta_v}[v\in S_1]$ be the probability that $v$ is active in one time step and $\ap(v\cond\bar{u})=\Pr_{\mathcal{D}_v,\theta_v}[v\in S_1\mid u\notin S_0]$ be the corresponding conditional probability.}

\paragraph{Social network, diffusion model and influence maximization}
A social network is modeled as a weighted directed graph $G=(V,E,p)$, where $V$ is the set of $|V|=n$ nodes and $E$ is the set of directed edges.
Each edge $(u,v)\in E$ is associated with a weight or probability $p_{uv}\in[0, 1]$.
For convenience, we assume that $p_{uv}=0$ if $(u,v)\notin E$ and $p_{uv}>0$ otherwise.
We also use $N(v)=N^{in}(v)$ to denote the in-neighbors of node $v\in V$.

The information or influence propagates through the network in discrete time steps.
Each node $v\in V$ is either \emph{active} or \emph{inactive}, indicating whether it receives the information.
Denote by $S_{\tau}\subseteq V$ the set of active nodes at time step $\tau$.
The nodes in the set $S_0$ at time step $0$ are called \emph{seeds}.
The diffusion is assumed to be \emph{progressive}, which means a node will remain active once it is activated.
Thus, for all $\tau\geq 1$, $S_{\tau-1}\subseteq S_{\tau}$.

Given a set of seeds $S_0$, the diffusion model describes how the information propagates and $S_{\tau}$ is generated for each $\tau\geq 1$.
In the literature, the two most well-known diffusion models are the \emph{independent cascade} (IC) model and the \emph{linear threshold} (LT) model \citep{KempeKT03}.

In the IC model, at time step $\tau$, initially $S_{\tau}=S_{\tau-1}$.
Next, for each node $v\notin S_{\tau-1}$, each node $u\in N(v)\cap(S_{\tau-1}\setminus S_{\tau-2})$ will try to activate $v$ \emph{independently} with probability $p_{uv}$ (denote $S_{-1}=\emptyset$).
Thus, $v$ becomes active with probability $1-\prod_{u\in N(v)\cap(S_{\tau-1}\setminus S_{\tau-2})}(1-p_{uv})$ at this step.
Once activated, $v$ will be added into $S_{\tau}$.
The propagation terminates when at the end of some time step $\tau$, $S_{\tau}=S_{\tau-1}$.
Clearly, the process proceeds in at most $n-1$ time steps and we use $(S_0,S_1,\cdots,S_{n-1})$ to denote the random sequence of the active nodes.

In the LT model, following the convention, we use $w$ instead of $p$ to denote the edge weight vector.
In this model, each node $v\in V$ needs to satisfies the \emph{normalization condition} that $\sum_{u\in N(v)} w_{uv}\leq 1$.
Besides, each $v\in V$ is associated with a threshold $r_v$, which is sampled independently and uniformly from $[0,1]$ before the diffusion starts.
During the diffusion process, at time step $\tau$, initially $S_{\tau}=S_{\tau-1}$.
Then, for each node $v\notin S_{\tau-1}$, it will be added into $S_{\tau}$ if $\sum_{u\in S_{\tau-1}\cap N(v)} w_{uv}\geq r_v$.
The diffusion also terminates if $S_{\tau}=S_{\tau-1}$ for some time step $\tau$ and we use $(S_0,S_1,\cdots,S_{n-1})$ to denote the random sequence of the active nodes.

Let $\Phi(S_0)=S_{n-1}$ be the final active set given the seed set $S_0$.
Its expected size is denoted by $\E[|\Phi(S_0)|]$
and is commonly called the {\em influence spread} of $S_0$.
In general, the \emph{influence spread function} $\sigma:2^V\rightarrow\mathbb{R}_{\ge 0}$ is defined as $\sigma(S)=\E[|\Phi(S)|]$ for any $S\subseteq V$.
Sometimes, we use $\sigma^p(\cdot)$ or $\sigma^w(\cdot)$ to specify the graph parameters explicitly.
\emph{Influence maximization} (IM) asks to find a set of at most $k$ seeds so as to maximize the influence spread of the chosen seed set.
Formally, under a specific diffusion model (such as IC or LT models), given a positive integer $k\leq n$, influence maximization corresponds to the following problem: $\argmax_{S\subseteq V, |S|\le k} \sigma(S)$.

%A realization of the above random process can be described by a sequence of active nodes $(S_0,S_1,\cdots,S_{n-1})$, referred to as a \emph{cascade}.

%Equivalently, the IC model can be defined via \emph{living-edge graph}.
%Let $\mathcal{G}$ be a random graph such that an edge $(u,v)$ appears with probability $p_{uv}$ independently.
%Fixing a graph $G$ drawn from $\mathcal{G}$ and a seed set $S$, let $\sigma(G,S)$ be the nodes reachable from $S$ in $G$.
%Then $f(S)=\sum_{G\in\mathcal{G}}\Pr[G]|\sigma(G,S)|$.

\paragraph{The sampling setting}
Standard influence maximization problem takes as input the social network $G=(V,E,p)$, based on which one can compute or estimate the influence spread function $\sigma$.
In this paper, we consider the problem in the sampling setting where $G$ is not given explicitly.

A \emph{cascade} refers to a realization of the sequence of the active nodes $(S_0,S_1,\cdots,S_{n-1})$.
By slightly abusing the notation, we still denote the cascade by $(S_0,S_1,\cdots,S_{n-1})$.
In the sampling setting, a set of $t$ \emph{independent} cascades $(S_{i,0},S_{i,1},\cdots,S_{i,n-1})_{i=1}^t$ is given as input, where the seed set $S_{i,0}$ in cascade $i$ is generated \emph{independently} from a \emph{seed set distribution} $\mathcal{D}$ over the node sets, and given $S_{i,0}$, the sequence $(S_{i,1},\cdots,S_{i,n-1})$ is generated according to the specified diffusion rules.
Throughout this work, we assume that $\mathcal{D}$ is a \emph{product distribution}; in other words, each node $u\in V$ is drawn as a seed \emph{independently}.
We aim to solve the following two problems.

\begin{enumerate}
	\item \textbf{Network inference\footnote{In the literature, network inference often means to recover network structure, namely the edge set $E$. Here we slightly abuse the terminology to also mean learning edge parameters.}.} Given a set of $t$ samples $(S_{i,0},S_{i,1},\cdots,S_{i,n-1})_{i=1}^t$ defined as above, estimate the values of $p_{uv}$ within an additive error. More formally, for some $\varepsilon\in(0,1)$, compute a vector $\hat{p}$ such that $|\hat{p}_{uv}-p_{uv}|\leq\varepsilon$ for all $u,v\in V$.
	\item \sloppy\textbf{Influence maximization from samples (IMS).}
	Given a set of $t$ samples $(S_{i,0},S_{i,1},\cdots,S_{i,n-1})_{i=1}^t$ defined as above, find a set $\ALG$ of at most $k$ seeds such that $\sigma(\ALG)\geq \kappa\cdot \sigma(\OPT)$ for some constant $\kappa\in(0,1)$, where $\OPT$ denotes the optimal solution.
	%	\wei{I feel that notations $\ALG$ and $\OPT$ are a bit awkward representing sets. $\OPT$ is used often to represent the optimal objective value, not the solution set. How about changing them to $S^A$ and $S^*$? $S^*$ is often used as the optimal seed set for influence maximization. Zhijie, if you are fine with this change, you can do a global replacement. Thanks.
	%	}
\end{enumerate}

\paragraph{Notations}
Our algorithms actually only use $S_{i,0}$ and $S_{i,1}$ in those cascades to infer information about the graph, and we find it convenient to define some corresponding concepts and notations in advance.
These concepts are indeed crucial to our algorithm design.
For node $v\in V$, we use $q_v=\Pr_{\mathcal{D}}[v\in S_0]$ to denote the probability that $v$ is drawn as a seed.
We denote by $\ap_{G,\mathcal{D}}(v)$ the \emph{activation probability} of node $v$ \emph{in one time step} during a cascade $(S_0,S_1,\cdots,S_{n-1})$ over network $G$ when $S_0$ is drawn from the distribution $\mathcal{D}$.
Thus, $\ap_{G,\mathcal{D}}(v)=\Pr_{G,\mathcal{D}}[v\in S_1]$.
Note that it contains the possibility that $v$ itself is a seed, namely $v\in S_0\subseteq S_1$.
%Note that $ap_{G,\mathcal{D}}(v)$ depends on the randomness of both $S_0$ and the graph $G$.
For $u,v\in V$, we define $\ap_{G,\mathcal{D}}(v \cond u)=\Pr_{G,\mathcal{D}}[v\in S_1 \cond u\in S_0]$ and $\ap_{G,\mathcal{D}}(v\cond \bar{u})=\Pr_{G,\mathcal{D}}[v\in S_1 \cond u\notin S_0]$, respectively, which are the corresponding probabilities conditioned on whether $u$ is selected as a seed.
When the context is clear, we often omit the subscripts $G$ and $\mathcal{D}$ in the notation.

\paragraph{Chernoff-type bounds} Following are Chernoff-type bounds we will use in our analysis.

\begin{lemma}[Multiplicative Chernoff bound, \citealt{MitzenmacherU05}]
	\label{lemma: multiplicative chernoff bound}
	\sloppy Let $X_1, X_2, \cdots, X_n$ be independent random variables in $\{0, 1\}$ with $\Pr[X_i = 1] = p_i$.
	Let $X = \sum_{i=1}^{n} X_i$ and $\mu = \sum_{i=1}^{n}p_i$.
	Then, for $0 < a < 1$,
	\[ \Pr[X \geq (1+a) \mu] \leq e^{-\mu a^2/3}, \]
	and
	\[ \Pr[X \leq (1-a) \mu] \leq e^{-\mu a^2/2}. \]
\end{lemma}

\begin{lemma}[Additive Chernoff bound, \citealt{AlonS08}]
	\label{lemma: additive chernoff bound}
	Let $X_1,\cdots,X_n$ be independent random variables in $\{0,1\}$ with $\Pr[X_i=1]=p_i$.
	Let $X=\sum_{i=1}^{n}X_i$ and $\mu=\sum_{i=1}^{n}p_i$.
	Then for any $a>0$, we have
	\[ \Pr[X-\mu\geq a]\leq\exp(-a\min(1/5,a/4\mu)). \]
	Moreover, for any $a>0$, we have
	\[ \Pr[X-\mu\leq-a]\leq\exp(-a^2/2\mu). \]
\end{lemma}

\section{Algorithms under the IC Model}
\label{section: the independent cascade model}

In this section, we solve both network inference and IMS under the IC model.
In Section \ref{section: edge probabilities estimation}, we present network inference algorithms to estimate each edge probability within a small additive error.
In Section \ref{section: influence maximization from samples}, we present several IMS algorithms.

\subsection{Network Inference}
\label{section: edge probabilities estimation}

In this section, we present a novel algorithm under the IC model for estimating the edge probabilities of the underlying graph $G$, namely we need to find an estimate $\hat{p}$ of $p$ such that $|\hat{p}_{uv}-p_{uv}|\leq \varepsilon$ for all $u,v\in V$.
%Recall that $\\ap(v)$ denotes the activation probability of $v$ in one time step by a random seed set, and $\\ap(v\cond\bar{u})$ denotes the activation probability of $v$ in one time step conditioned on that $u$ is not a seed.
While all previous studies rely on the maximum likelihood estimation to estimate $\hat{p}$ \citep{NetrapalliS12,NarasimhanPS15,Pouget-AbadieH15}, 
our algorithm is based on
the following key observation on the connection between $p_{uv}$ and the one-step activation probabilities
$\ap(v)$ and $\ap(v\cond\bar{u})$.
%$\\ap(v)=\\ap(v\cond \bar{u})+(1-\\ap(v\cond\bar{u}))q_up_{uv}$.
We remark that our algorithm does not rely on the knowledge of edges in graph $G$, and in fact it can be used to also reconstruct
the edges of the graph.

\begin{lemma}  \label{lem:puv}
	Under the IC model, for any $u,v\in V$ with $u\not=v$,
	\[ p_{uv}=\frac{\ap(v)-\ap(v\cond\bar{u})}{q_u(1-\ap(v\cond\bar{u}))}. \]
\end{lemma}

\begin{proof}
	To avoid confusion, we write the underlying graph $G$ and the seed distribution $\mathcal{D}$ explicitly in notation $\ap(\cdot)$, namely $\ap(v)=\ap_{G,\mathcal{D}}(v)$.
	Consider the subgraph $G'=G \setminus \{u\}$ by removing node $u$.
	Node $v$ has two chances to be activated \emph{in one time step}: either by nodes in $G'$ (including the case where $v$ itself is a seed) or by node $u$.
	Since $\mathcal{D}$ is a product distribution, we have
	\[ \ap_{G,\mathcal{D}}(v)=\ap_{G',\mathcal{D}}(v)+(1-\ap_{G',\mathcal{D}}(v))q_up_{uv}. \]
	Besides, $\ap_{G',\mathcal{D}}(v)=\ap_{G,\mathcal{D}}(v\cond\bar{u})$ since when considering one-step activation of $v$, node $u$ not being the seed is equivalent to removing it from the graph.
	Plugging the equality into the last one, we obtain
	\[ p_{uv}=\frac{\ap_{G,\mathcal{D}}(v)-\ap_{G,\mathcal{D}}(v\cond\bar{u})}{q_u(1-\ap_{G,\mathcal{D}}(v\cond\bar{u}))}, \]
	which proves the lemma.
\end{proof}

Equipped with Lemma \ref{lem:puv}, we are able to estimate $p_{uv}$ by estimating $q_u,\ap(v)$ and $\ap(v\cond\bar{u})$ respectively from cascade samples.
Let $t_u=|\{i\in[t]\mid u\in S_{i,0} \}|$ be the number of cascades where $u$ is a seed, $t_{\bar{u}}=|\{i\in[t]\mid u\notin S_{i,0} \}|$ the number of cascades where $u$ is not a seed, $t^v=|\{i\in[t]\mid v\in S_{i,1}\}|$ the number of cascades where $v$ is activated in one time step and $t_{\bar{u}}^v=|\{i\in[t]\mid u\notin S_{i,0},v\in S_{i,1}\}|$ the number of cascades where $u$ is not a seed and $v$ is activated in one time step.
Then, $\hat{q}_u=t_u/t, \widehat{\ap}(v)=t^v/t$ and $\widehat{\ap}(v\cond \bar{u})=t_{\bar{u}}^v/t_{\bar{u}}$ are good estimates of $q_u,\ap(v)$ and $\ap(v\cond\bar{u})$, respectively.
The formal procedure is formulated as Algorithm \ref{algorithm: estimate edge probabilities}.

Algorithm \ref{algorithm: estimate edge probabilities} needs to work under Assumption \ref{assumption: estimate edge probabilities} below, which ensures that all quantities are well estimated.
Assumption \ref{assumption: estimate edge probabilities} consists of two conditions.
The first means that node $v\in V$ has a non-negligible probability of not being activated in one time step.
The second means that the probability of a node $u\in V$ being selected as a seed is neither too low nor too high.

\begin{algorithm}[tb]
	\caption{Estimate Edge Probabilities}
	\label{algorithm: estimate edge probabilities}
	\begin{algorithmic}[1]
		\REQUIRE A set of cascades $(S_{i,0},S_{i,1},\cdots,S_{i,n-1})_{i=1}^t$.
		\ENSURE $\{\hat{p}_{uv}\}_{u,v\in V}$ such that $|\hat{p}_{uv}-p_{uv}|\leq\varepsilon$ for all $u,v\in V$.
		\FOR{each $u\in V$}
		\STATE Compute $\hat{q}_u=t_u/t$, where $t_u=|\{i\in[t]\mid u\in S_{i,0} \}|$.
		\ENDFOR
		\FOR{each $v\in V$}
		\STATE Compute $\widehat{\ap}(v)=t^v/t$, where $t^v=|\{i\in[t]\mid v\in S_{i,1}\}|$.
		\ENDFOR
		\FOR{each $v\in V$}
		\FOR{each $u\in V$}
		\STATE Compute $\widehat{\ap}(v\cond \bar{u})=t_{\bar{u}}^v/t_{\bar{u}}$, where $t_{\bar{u}}=|\{i\in[t]\mid u\notin S_{i,0} \}|$ and $t_{\bar{u}}^v=|\{i\in[t]\mid u\notin S_{i,0},v\in S_{i,1}\}|$.
		\STATE Let $\hat{p}_{uv}=\frac{\widehat{\ap}(v)-\widehat{\ap}(v\cond\bar{u})}{\hat{q}_u(1-\widehat{\ap}(v\cond\bar{u}))}$.
		\ENDFOR
		\ENDFOR
		\STATE \textbf{return} $\{\hat{p}_{uv}\}_{u,v\in V}$.
	\end{algorithmic}
\end{algorithm}

%\begin{algorithm}[t]
%	\SetAlgoLined
%	\KwIn{A set of independent cascades $(S_{i,0},S_{i,1},\cdots,S_{i,n-1})_{i=1}^t$.}
%	%		\wei{Do we use $\varepsilon$ in the following algorithm? If not,
%	%		we do not need to list $\varepsilon$ as the input parameter.}
%	\KwOut{$\{\hat{p}_{uv}\}_{u,v\in V}$ such that $|\hat{p}_{uv}-p_{uv}|\leq\varepsilon$ for all $u,v\in V$.}
%	
%	\ForEach{$u\in V$}{
%		Compute $\hat{q}_u=t_u/t$, where $t_u=|\{i\in[t]\mid u\in S_{i,0} \}|$.
%	}
%	
%	\ForEach{$v\in V$}{
%		Compute $\widehat{\\ap}(v)=t^v/t$, where $t^v=|\{i\in[t]\mid v\in S_{i,1}\}|$.
%	}
%	
%	\ForEach{$v\in V$}{
%		\ForEach{$u\in V$}{
%			Compute $\widehat{\\ap}(v\cond \bar{u})=t_{\bar{u}}^v/t_{\bar{u}}$, where $t_{\bar{u}}=|\{i\in[t]\mid u\notin S_{i,0} \}|$ and $t_{\bar{u}}^v=|\{i\in[t]\mid u\notin S_{i,0},v\in S_{i,1}\}|$.
%			
%			Let $\hat{p}_{uv}=\frac{\widehat{\\ap}(v)-\widehat{\\ap}(v\cond\bar{u})}{\hat{q}_u(1-\widehat{\\ap}(v\cond\bar{u}))}$.
%		}
%	}
%	
%	\Return $\{\hat{p}_{uv}\}_{u,v\in V}$.
%	
%	\caption{Estimate Edge Probabilities}
%%	\label{algorithm: estimate edge probabilities}
%\end{algorithm}

\begin{assumption}[Edge probabilities estimation under the IC model]
	\label{assumption: estimate edge probabilities}
	%	The following are the assumptions for estimating edge probabilities.
	For some parameters $\alpha\in(0,1],\gamma\in(0,1/2]$,
	\begin{enumerate}
		\item $\ap(v)\leq 1-\alpha$ for all $v\in V$.
		\item $\gamma\leq q_u\leq 1-\gamma$ for all $u\in V$.
	\end{enumerate}
\end{assumption}

We now give an analysis of Algorithm \ref{algorithm: estimate edge probabilities}.
Lemma \ref{lemma: estimate ap() in learning edge probabilities} below gives the number of samples we need to estimate $q_u,\ap(v)$ and $\ap(v\cond\bar{u})$ within a small accuracy.

\begin{lemma}
	\label{lemma: estimate ap() in learning edge probabilities}
	Under Assumption \ref{assumption: estimate edge probabilities}, for any $\eta\in(0,4/5),\delta\in(0,1)$, for $\hat{q}_u$, $\widehat{\ap}(v)$, and $\widehat{\ap}(v\cond\bar{u})$ defined in Algorithm~\ref{algorithm: estimate edge probabilities},
	if the number of samples $t\geq\frac{16}{\gamma\eta^2}\ln\frac{12n}{\delta}$, with probability at least $1-\delta$, we have
	\begin{enumerate}
		%		\item $\forall i, t_{u_i}-tq_i\leq \gamma t/2$.
		\item $|\hat{q}_u-q_u|\leq\eta q_u$ for all $u\in V$,
		\item $|\widehat{\ap}(v)-\ap(v)|\leq \eta$ for all $v\in V$,
		\item $|\widehat{\ap}(v\cond\bar{u})-\ap(v\cond\bar{u})|\leq\eta$ for all $u,v\in V$.
	\end{enumerate}
\end{lemma}

\begin{proof}
	For a node $u\in V$, for $i\in[t]$, let $X_i$ be a $0$-$1$ random variable such that $X_i=1$ iff $u\in S_{i,0}$.
	Thus $\hat{q}_u=\sum_{i=1}^{t}X_i/t$.
	By Lemma \ref{lemma: multiplicative chernoff bound},
	\[ \Pr[|\hat{q}_u-q_u|\geq\eta q_u]=\Pr[|\sum_{i=1}^{t}X_i-tq_u|\geq \eta tq_u]\leq2\exp(-tq_u\eta^2/3)\leq2\exp(-t\gamma\eta^2/3)\leq\delta/(3n). \]
	The last inequality requires that $t\geq\frac{3}{\gamma\eta^2}\ln\frac{6n}{\delta}$.
	
	For a node $v\in V$, for $i\in[t]$, let $Y_i$ be a $0$-$1$ random variable such that $Y_i=1$ iff $v\in S_{i,1}$.
	Thus $\widehat{\ap}(v)=\sum_{i=1}^{t}Y_i/t$.
	By Lemma \ref{lemma: additive chernoff bound} ($a=\eta t,\mu=t\cdot \ap(v)$),
	\begin{align*}
		\Pr[|\widehat{\ap}(v)-\ap(v)|\geq\eta] &\leq\exp(-\eta t\min(1/5,\eta/(4\cdot \ap(v))))+\exp(-\eta^2 t/(2\cdot \ap(v))) \\
		&\leq\exp(-\eta t\min(1/5,\eta/4))+\exp(-\eta^2 t/2) \\
		&\leq\delta/(6n)+\delta/(6n)=\delta/(3n).
	\end{align*}
	The second inequality holds since $\ap(v)\leq 1$ by definition.
	The last inequality requires that $t\geq\frac{4}{\eta^2}\ln\frac{6n}{\delta}$.
	
	%	For any $i$, since $\E[t_{u_i}]=tq_i$, by Lemma \ref{lemma: additive chernoff bound},
	%	\[ \Pr[t_{u_i}-tq_i\geq\gamma t/2]\leq\exp(-(\gamma t/2)\min(1/5,\gamma/8q_i))\leq\delta/(4n). \]
	%	Since $\E[\hat{q}_i]=q_i$, by Lemma \ref{lemma: multiplicative chernoff bound}, $\Pr[|\hat{q}_i-q_i|\geq\eta q_i]\leq 2\exp(-tq_i\eta^2/3)\leq 2\exp(-t\gamma\eta^2/3)\leq\delta/(4n)$.
	%
	%	
	%	%	For any $i$, if $u_i$ is the only neighbor of $v$, then $\widehat{\\ap}_{G\setminus\{u_i\}}(v)=ap_{G\setminus\{u_i\}}(v)=0$.
	%	%	If $v$ has another neighbor $u_j$, then $ap_{G\setminus\{u_i\}}(v)\geq q_jp_j\geq\beta\gamma$.
	%	Since $t-t_{u_i}\geq t-tq_i-\gamma t/2\geq t-(1-\gamma)t-\gamma t/2=\gamma t/2$, we have
	%	\begin{align*}
	%		&\Pr[|\widehat{\\ap}_{G\setminus\{u_i\}}(v)-ap_{G\setminus\{u_i\}}(v)|\geq\eta] \\
	%		&\leq \exp(-(t-t_{u_i})\eta\min(1/5,\eta/4ap_{G\setminus\{u_i\}}(v)))+\exp(-(t-t_{u_i})\eta^2/2ap_{G\setminus\{u_i\}}(v)) \\
	%		&\leq \exp(-(\gamma\eta t/2)\min(1/5,\eta/4ap_{G\setminus\{u_i\}}(v)))+\exp(-\gamma\eta^2 t/4ap_{G\setminus\{u_i\}}(v)) \\
	%		&\leq\delta/(8n)+\delta/(8n)\leq\delta/(4n)
	%	\end{align*}
	%	
	%	By union bound, these bad events occur with probability at most $1-\delta$.
	
	For $u\in V$, let $t_{\bar{u}}=|\{i\in[t]\mid u\notin S_{i,0} \}|$ be the number of cascades where $u$ is not a seed.
	Since $q_u\leq1-\gamma$, $8\ln(12n/\delta)/\eta^2\leq t(1-q_u)/2$.
	By Lemma \ref{lemma: multiplicative chernoff bound} ($a=1/2,\mu=t(1-q_u)$), 
	%	\wei{According to Lemma \ref{lemma: multiplicative chernoff bound}, the $12$ below could be changed to $8$. Do we need this $8$ to get the $n^2$ in the end?}
	\[ \Pr[t_{\bar{u}}\leq 8\ln(12n/\delta)/\eta^2]\leq\Pr[t_{\bar{u}}\leq t(1-q_u)/2]\leq \exp(-t(1-q_u)/8)\leq\exp(-t\gamma/8)\leq\delta/(6n^2). \]
	The last inequality holds as long as $t\geq 16\ln(6n/\delta)/\gamma$.
	
	Given a fixed $\ell$, assume that there are $t_{\bar{u}}=\ell$ cascades where $u$ is not a seed.
	For $i\in[\ell]$, let $Z_i$ be a $0$-$1$ random variable such that $Z_i=1$ iff $v\in S_{i,1}$ in the $i$-th cascade, among the $\ell$ cascades where $u$ is not a seed.
	Then, $\widehat{\ap}(v\cond \bar{u})=\sum_{i=1}^{\ell}Z_i/\ell$.
	By Lemma \ref{lemma: additive chernoff bound} ($a=\eta\ell, \mu=\ell\cdot \ap(v\cond \bar{u})$),
	\begin{align*}
		&\Pr[|\widehat{\ap}(v\cond \bar{u})-\ap(v\cond \bar{u})| \geq \eta \cond t_u=\ell] \\ &\leq\exp(-\eta\ell\min(1/5,\eta/(4\cdot \ap(v\cond \bar{u}))))+\exp(-\eta^2\ell/(2\cdot \ap(v\cond \bar{u}))) \\
		& \leq\exp(-\eta\ell\min(1/5,\eta/4))+\exp(-\eta^2\ell/2).
	\end{align*}
	The last inequality holds since $\ap(v\cond u)\leq 1$ by definition.
	If $\ell\geq 8\ln(12n/\delta)/\eta^2$, then
	\[ \Pr[|\widehat{\ap}(v\cond \bar{u})-\ap(v\cond \bar{u})|\geq \eta \cond t_u=\ell]\leq\delta/(6n^2). \]
	
	By law of total probability,
	\begin{align*}
		&\Pr[|\widehat{\ap}(v\cond \bar{u})-\ap(v\cond \bar{u})|\geq \eta] \\
		&\leq \left(\sum_{\ell<8\ln(12n/\delta)/\eta^2}+\sum_{\ell\geq 8\ln(12n/\delta)/\eta^2}\right)\Pr[t_u=\ell]\Pr[|\widehat{\ap}(v\cond \bar{u})-\ap(v\cond \bar{u})|\geq \eta\mid t_u=\ell] \\
		&\leq \Pr[t_{\bar{u}}<8\ln(12n/\delta)/\eta^2]+\Pr[t_{\bar{u}}\geq 8\ln(12n/\delta)/\eta^2]\cdot\delta/(6n^2) \\
		&\leq \delta/(6n^2)+\delta/(6n^2)=\delta/(3n^2).
	\end{align*}
	
	The lemma follows immediately by union bound.
\end{proof}

As stated below, Theorem \ref{theorem: estimate edge probabilities} derives a theoretical guarantee for Algorithm \ref{algorithm: estimate edge probabilities}.

\begin{theorem}
	\label{theorem: estimate edge probabilities}
	Under Assumption \ref{assumption: estimate edge probabilities}, for any $\varepsilon,\delta\in(0,1)$, let $\eta=\varepsilon\alpha\gamma/4<1/8$, and $\{\hat{p}_{uv}\}_{u,v\in V}$ be the set of edge probabilities returned by Algorithm \ref{algorithm: estimate edge probabilities}.
	If the number of cascades $t\geq \frac{16}{\gamma\eta^2}\ln\frac{12n}{\delta}=\frac{256}{\varepsilon^2\alpha^2\gamma^3}\ln\frac{12n}{\delta}$, with probability at least $1-\delta$, for any $u,v\in V$, $|\hat{p}_{uv}-p_{uv}|\leq\varepsilon$.
\end{theorem}

\begin{proof}
	With probability at least $1-\delta$, all the events in Lemma \ref{lemma: estimate ap() in learning edge probabilities} occur.
	We assume that this is exactly the case in the following.
	Since $\ap(v\cond \bar{u})\leq \ap(v)\leq 1-\alpha$, we have $1-\ap(v\cond\bar{u})\geq\alpha$.
	By the value of $\eta$ and the assumption that $q_u\geq\gamma$, we have
	\begin{equation}  \label{eq:eta}
		\eta\leq\frac{\varepsilon\gamma}{4}(1-\ap(v\cond\bar{u}))\leq\frac{\varepsilon}{4}q_u(1-\ap(v\cond\bar{u})). 
	\end{equation}
	
	To prove $\hat{p}_{uv}\leq p_{uv}+\varepsilon$, we have
	\begin{align}
		\hat{p}_{uv} &=\frac{\widehat{\ap}(v)-\widehat{\ap}(v\cond\bar{u})}{\hat{q}_u(1-\widehat{\ap}(v\cond\bar{u}))} \nonumber \\
		&\leq\frac{\ap(v)-\ap(v\cond\bar{u})+2\eta}{(1-\eta)q_u(1-\ap(v\cond\bar{u})-\eta)} \nonumber \\
		&\leq\frac{\ap(v)-\ap(v\cond\bar{u})+2\eta}{(1-\eta)(1-\varepsilon\gamma/4)q_u(1-\ap(v\cond\bar{u}))} \nonumber \\
		&\leq \frac{p_{uv}+\varepsilon/2}{(1-\eta)(1-\varepsilon\gamma/4)}\leq p_{uv}+\varepsilon. \label{eq:puvepsilon}
	\end{align}
	The first inequality holds due to Lemma \ref{lemma: estimate ap() in learning edge probabilities}.
	The second inequality holds by applying the first inequality in Eq.~\eqref{eq:eta}.
	The third inequality holds due to Lemma \ref{lem:puv} and the second inequality in Eq.~\eqref{eq:eta}.
	To see the correctness of the last inequality, first observe that
	\begin{align*}
		&(p_{uv}+\varepsilon)(1-\eta)(1-\varepsilon\gamma/4) \\
		&\geq (p_{uv}+\varepsilon)(1-\eta-\varepsilon\gamma/4) \\
		&\geq (p_{uv}+\varepsilon)-(1+\epsilon)(\eta+\varepsilon\gamma/4).
	\end{align*}
	Next, note that
	\[ (1+\epsilon)(\eta+\varepsilon\gamma/4)=(1+\epsilon)(1+\alpha)\varepsilon\gamma/4\leq(1+\varepsilon)\varepsilon/4\leq\varepsilon/2. \]
	The equality is due to the definition of $\eta$.
	The two inequalities hold since $\alpha\in(0,1],\gamma\in(0,1/2]$ and $\varepsilon\in(0,1)$, respectively.
	Combining the above two observations, we have the desired inequality
	\[ (p_{uv}+\varepsilon)(1-\eta)(1-\varepsilon\gamma/4)\geq (p_{uv}+\varepsilon)-\varepsilon/2=p_{uv}+\varepsilon/2. \]
	%	\wei{How to derive the last inequality above? Not very immediate to me.}
	
	On the other hand, to prove $\hat{p}_{uv}\geq p_{uv}-\varepsilon$, first assume that $p_{uv}\geq\varepsilon$, since otherwise the claim would be trivial for $\hat{p}_{uv}\geq0$.
	We now have
	\begin{align*}
		\hat{p}_{uv}
		&=\frac{\widehat{\ap}(v)-\widehat{\ap}(v\cond\bar{u})}{\hat{q}_u(1-\widehat{\ap}(v\cond\bar{u}))} \\
		&\geq\frac{\ap(v)-\ap(v\cond\bar{u})-2\eta}{(1+\eta)q_u(1-\ap(v\cond\bar{u})+\eta)} \\
		&\geq\frac{ap_G(v)-\ap(v\cond \bar{u})-2\eta}{(1+\eta)(1+\varepsilon\gamma/4)q_u(1-\ap(v\cond\bar{u}))} \\
		&\geq \frac{p_{uv}-\varepsilon/2}{(1+\eta)(1+\varepsilon\gamma/4)}\geq p_{uv}-\varepsilon.
	\end{align*}
	The first inequality holds due to Lemma \ref{lemma: estimate ap() in learning edge probabilities}.
	The second inequality holds by applying the first inequality in Eq.~\eqref{eq:eta}.
	The third inequality holds due to Lemma \ref{lem:puv} and the second inequality in Eq.~\eqref{eq:eta}.
	The last inequality follows from a similar argument as 
	the one for the last inequality in Eq.~\eqref{eq:puvepsilon},
	and we omit it for conciseness.
	%\wei{Why writing $ap_G(v)-ap_{G\setminus\{u_i\}}(v)$? According to Lemma~\ref{lem:puv}, isn't writing $\\ap(v)-\\ap(v\cond\bar{u})$ more direct?
	%	Also, how to derive the last inequality?}
\end{proof}

With the ability of estimating edge probabilities, we further show that
we can recover the graph structure by a standard threshold approach \citep{NetrapalliS12,Pouget-AbadieH15}.
The formal procedure is depicted as Algorithm \ref{algorithm: recover network structure}, which estimates the edge probabilities to a prescribed accuracy and returns the edges whose estimated probabilities are above a prescribed threshold.
Its guarantee is shown in Theorem \ref{theorem: recover network structure}, which shows that no ``zero-probability edge'' is incorrectly recognized as an edge.
Besides, only small-probability edges are omitted, which is reasonable for practical use.

%\begin{assumption}[Assumptions for Recovering Network Structure]
%	\label{assumption: recover network structure}
%	For some parameters $\alpha,\beta\in(0,1],\gamma\in(0,1/2]$,
%	\begin{enumerate}
%		\item $p_{uv}>\beta$ for all $(u,v)\in E$.
%		\item $\\ap(v)\leq 1-\alpha$ for all $v\in V$.
%		\item $\gamma\leq q_u\leq 1-\gamma$ for all $u\in V$.
%	\end{enumerate}
%\end{assumption}

\begin{algorithm}[tb]
	\caption{Recover Network Structure}
	\label{algorithm: recover network structure}
	\begin{algorithmic}[1]
		\REQUIRE A set of cascades $(S_{i,0},S_{i,1},\cdots,S_{i,n-1})_{i=1}^t$, parameter $\beta\in (0,1)$.
		\ENSURE An estimated edge set $\hat{E}$.
		\STATE $\{\hat{p}_{uv}\}_{u,v\in V}=\mbox{Estimate-Edge-Probabilities}$\\$((S_{i,0},S_{i,1},\cdots,S_{i,n-1})_{i=1}^t)$.
		\COMMENT{ With estimation accuracy $\beta/2$.}
		\STATE \textbf{return} $\hat{E}=\{(u,v)\mid \hat{p}_{uv}>\beta/2\}$.
	\end{algorithmic}
\end{algorithm}

%\begin{restatable}{theorem}{thmRecoverStructure}
%	\label{theorem: recover network structure}
%	If the number of cascades $t\geq\frac{1024}{\alpha^2\beta^2\gamma^3}\ln\frac{4n}{\delta}$, with probability at least $1-\delta$, the edge set $\hat{E}$ returned by Algorithm \ref{algorithm: recover network structure} is exactly the true edge set $E$.
%\end{restatable}

\begin{theorem}
	\label{theorem: recover network structure}
	Under Assumption \ref{assumption: estimate edge probabilities}, if the number of cascades $t\geq\frac{1024}{\alpha^2\beta^2\gamma^3}\ln\frac{4n}{\delta}$, with probability at least $1-\delta$, the edge set $\hat{E}$ returned by Algorithm \ref{algorithm: recover network structure} satisfies (1) $\hat{E}\subseteq E$, and (2) if $p_{uv}>\beta$, then $(u,v)\in \hat{E}$.	
	As a corollary, if $p_{uv}>\beta$ for all $(u,v)\in E$, then $\hat{E}=E$.
\end{theorem}

\begin{proof}
	By Theorem \ref{theorem: estimate edge probabilities}, $|\hat{p}_{uv}-p_{uv}|\leq\varepsilon=\beta/2$ for all $u,v\in V$ with probability at least $1-\delta$.
	If $(u,v)\notin E$, then $p_{uv}=0$ and hence $\hat{p}_{uv}\leq\beta/2$, which implies $(u,v)\notin \hat{E}$.
	Thus, $\hat{E}\subseteq E$.
	If $(u,v)\in E$ and $p_{uv}>\beta$.
	Then, $\hat{p}_{uv}\geq p_{uv}-\beta/2>\beta/2$ and hence $(u,v)\in\hat{E}$.
	Finally, if $p_{uv}>\beta$ for all $(u,v)\in E$, then $E\subseteq \hat{E}$ and hence $\hat{E}=E$, which concludes the proof.
\end{proof}

\paragraph{Discussion}
\sloppy It is worth comparing the result in \citep{NetrapalliS12,NarasimhanPS15,Pouget-AbadieH15} with ours, 
since all of them studied network inference under the IC model.
Specifically, \citet{NetrapalliS12} initiated the study of recovering network structure and did not consider the estimation of edge parameters.
\citet{NarasimhanPS15} and \citet{Pouget-AbadieH15} studied how to estimate edge parameters.
Both of them used the Euclidean norm of the edge probability vector as the measurement of accuracy, while we use the infinite norm.
Besides, in \citep{NarasimhanPS15}, it was additionally assumed that the network structure is known in advance.
In \citep{Pouget-AbadieH15}, totally different assumptions were used, which seems incomparable to ours,
and thus we will not further compare against it below.

There are several important differences besides the above.
First, the approaches taken are different.
All the algorithms in the previous works build on the maximum likelihood estimation (MLE) and require to solve a convex program, while we directly find a closed-form expression for the edge probability $p_{uv}$, thus rendering fast implementation.

Second, the assumptions required are different.
The assumptions $p_{uv}>\beta$ for all $u,v\in V$ and $\gamma\leq q_u\leq 1-\gamma$ for all $u\in V$ are also required in the previous works (though may be presented in different forms).
The key difference is the condition $\ap(v)\leq 1-\alpha$ for all $v\in V$.
In \citep{NetrapalliS12}, its role is replaced by the \emph{correlation decay} condition, which requires that $\sum_{u\in N(v)} p_{uv}\leq 1-\alpha$ for all $v\in V$.
In \citep{NarasimhanPS15}, it is instead assumed that $1-\prod_{u\in N(v)}(1-p_{uv})\leq1-\alpha$ for all $v\in V$.
By observing that $\ap(v)\leq 1-\prod_{u\in N(v)}(1-p_{uv})\leq \sum_{u\in N(v)} p_{uv}$ (see the appendix), it is easy to see that our assumptions are the weakest compared with those in \citep{NetrapalliS12,NarasimhanPS15}.
Besides, Assumption \ref{assumption: estimate edge probabilities} enjoys the advantage that it is verifiable, since one can find suitable values for $\alpha$ and $\gamma$ by estimating $\ap(v)$ and $q_u$ from cascade samples.
However, it is impossible to verify the assumptions in \citep{NetrapalliS12,NarasimhanPS15} based only on cascade samples.
We remark that our network inference algorithm replies on the assumption that each seed node is independently sampled. 
This assumption is also made in \citep{NetrapalliS12,NarasimhanPS15} for the MLE method, but conceptually it might be easier to
relax this assumption with MLE.
We leave the relaxation of the independence sampling assumption of our method as a future work.

Finally, our algorithm has lower sample complexity compared with those in \citep{NetrapalliS12,NarasimhanPS15}.
Assume that $\ap(v)\leq 1-\prod_{u\in N(v)}(1-p_{uv})\leq \sum_{u\in N(v)} p_{uv}\leq 1-\alpha$.
Then, \citet{NetrapalliS12} needs $\tilde{O}(\frac{1}{\alpha^7\beta^2\gamma}D^2\log\frac{n}{\delta})$ samples to recover network structure, where $D$ is the maximum in-degree of the network, while we only need $O(\frac{1}{\alpha^2\beta^2\gamma^3}\ln\frac{4n}{\delta})$ samples by Theorem \ref{theorem: recover network structure}.
On the other hand, assume that the network structure is known and $m=|E|$.
\cite{NarasimhanPS15} needs $\tilde{O}(\frac{1}{\varepsilon^2\alpha^2\beta^2\gamma^2(1-\gamma)^4}mn\ln\frac{n}{\delta})$ samples to achieve $\|\hat{p}-p\|_2^2\leq\varepsilon$, while we only need $O(\frac{1}{\varepsilon\alpha^2\gamma^2}m\ln\frac{n}{\delta})$ samples by achieving $|\hat{p}_{uv}-p_{uv}|\leq\sqrt{\frac{\varepsilon}{m}}$.

\subsection{Influence Maximization from Samples}
\label{section: influence maximization from samples}

In this section, we present several IMS algorithms under the IC model.
In Section \ref{subsection: influence maximization under assumption estimate edge probabilities}, we present an \emph{approximation-preserving} algorithm under Assumption \ref{assumption: estimate edge probabilities}.
In Section \ref{subsection: influence maximization under assumptions independent of the network}, we show that under an alternative assumption (Assumption \ref{assumption: influence maximization under assumption independent of graph}), there is a \emph{constant} approximation algorithm for the problem.
An attractive feature of Assumption \ref{assumption: influence maximization under assumption independent of graph} 
(compared to Assumption \ref{assumption: estimate edge probabilities}) is that
it has no requirement on the network.
We also show that by slightly strengthening Assumption \ref{assumption: influence maximization under assumption independent of graph}, we again obtain an \emph{approximation-preserving} algorithm.

\subsubsection{IMS under Assumption \ref{assumption: estimate edge probabilities}}
\label{subsection: influence maximization under assumption estimate edge probabilities}

Our first IMS algorithm is presented as Algorithm \ref{algorithm: influence maximization under edge probabilities estimation assumption}.
It follows the canonical learning-and-then-optimization approach by first learning a surrogate graph $\hat{G}=(V,E,\hat{p})$ from the cascades and then executing any $\kappa$-approximation algorithm $A$ for standard influence maximization on $\hat{G}$ to obtain a solution as output.
The construction of $\hat{G}$ builds on Algorithm \ref{algorithm: estimate edge probabilities} and is obtained by estimating all the edge probabilities to a sufficiently small additive error.
Algorithm \ref{algorithm: influence maximization under edge probabilities estimation assumption} works under Assumption \ref{assumption: estimate edge probabilities}, since Algorithm \ref{algorithm: estimate edge probabilities} does.

\begin{algorithm}[tb]
	\caption{IMS-IC under Assumption \ref{assumption: estimate edge probabilities}}
	\label{algorithm: influence maximization under edge probabilities estimation assumption}
	\begin{algorithmic}[1]
		\REQUIRE A set of cascades $(S_{i,0},S_{i,1},\cdots,S_{i,n-1})_{i=1}^t$ and $k\in\mathbb{N}_+$.
		\STATE $\{\hat{p}_{uv}\}_{u,v\in V}=\mbox{Estimate-Edge-Probabilities}$\\$((S_{i,0},S_{i,1},\cdots,S_{i,n-1})_{i=1}^t)$.
		\COMMENT{ With estimation accuracy $\varepsilon k/(2n^3)$.}
		\STATE Let $\hat{G}=(V,E,\hat{p})$.
		\STATE Let $\ALG=A(\hat{G},k)$, where $A$ is a $\kappa$-approximation IM algorithm.
		\STATE \textbf{return} $\ALG$.
	\end{algorithmic}
\end{algorithm}

%\begin{algorithm}[t]
%	\SetAlgoLined
%	\KwIn{A set of independent cascades $(S_{i,0},S_{i,1},\cdots,S_{i,n-1})_{i=1}^t$ and $k\in\mathbb{N}_+$.}
%	
%	$\{\hat{p}_{uv}\}_{u,v\in V}=\mbox{Estimate-Edge-Probability}((S_{i,0},S_{i,1},\cdots,S_{i,n-1})_{i=1}^t)$, with estimation accuracy being $\varepsilon k/(2n^3)$.
%	%	\wei{Similar to my comment in Algorithm 3, is $\varepsilon k/(2n^3)$ a parameter we really pass into algorithm Estimate-Edge-Probability?
%	%		It looks like we only require $t$ to be lower bounded by some value related to $\varepsilon$ to make it work.
%	%	}
%	
%	Let $\hat{G}=(V,E,\hat{p})$.
%	
%	Let $\ALG=A(\hat{G},k)$, where $A$ is a $\kappa$-approximation algorithm for influence maximization.
%	
%	\Return $\ALG$.
%	
%	\caption{Influence Maximization from Samples under Assumption \ref{assumption: estimate edge probabilities}}
%%	\label{algorithm: influence maximization under edge probabilities estimation assumption}
%\end{algorithm}

The correctness of Algorithm \ref{algorithm: influence maximization under edge probabilities estimation assumption} relies on Lemma \ref{lemma: estimate influence}, which translates the estimation error in edge probabilities into the learning error in the influence spread function.
%which shows that as long as the estimation of edge probabilities is good enough, the influence spread function corresponding to the surrogate graph is very close to the one corresponding to the original graph.
We use it in Theorem \ref{theorem: influence maximization under edge probabilities estimation assumption} to prove that with high probability, Algorithm \ref{algorithm: influence maximization under edge probabilities estimation assumption} almost preserves the approximation ratio of any standard influence maximization algorithm $A$.

\begin{lemma}[\citealt{NarasimhanPS15}]
	\label{lemma: estimate influence}
	Fix $S\subseteq V$.
	Under the IC model, for any two edge probability vectors $p,\hat{p}$ with $\|p-\hat{p}\|_1\leq\varepsilon/n$, we have $|\sigma^p(S)-\sigma^{\hat{p}}(S)|\leq \varepsilon$.
\end{lemma}

\begin{theorem}
	\label{theorem: influence maximization under edge probabilities estimation assumption}
	Under Assumption \ref{assumption: estimate edge probabilities}, for any $\varepsilon\in (0,1)$ and $k \in \mathbb{N}_+$, 
	suppose that the number of cascades $t\geq\frac{1024}{\varepsilon^2\alpha^2\gamma^3}\frac{n^6}{k^2}\ln\frac{12n}{\delta}$.
	Let $A$ be a $\kappa$-approximation algorithm for influence maximization.
	Let $\ALG$ be the set returned by Algorithm \ref{algorithm: influence maximization under edge probabilities estimation assumption} and $\OPT$ be the optimal solution on the original graph.
	We have
	\[ \Pr[\sigma(\ALG)\geq(\kappa-\varepsilon)\sigma(\OPT)]\geq1-\delta. \]
\end{theorem}

\begin{proof}
	By Theorem \ref{theorem: estimate edge probabilities}, with probability at least $1-\delta$, for any $u,v\in V$, $|\hat{p}_{uv}-p_{uv}|\leq\varepsilon k/(2n^3)$.
	Hence, $\|p-\hat{p}\|_1=\sum_{u,v\in V}|p_{uv}-\hat{p}_{uv}|\leq\varepsilon k/(2n)$.
	Applying this condition to Lemma~\ref{lemma: estimate influence}, we have that $|\sigma^p(S)-\sigma^{\hat{p}}(S)|\leq \varepsilon k/2$ for every seed set $S$.
	We thus have
	%	\wei{fixed some notation below. First, $\sigma_v$ should be replaced with $\sigma$; second, $\hat{\sigma}$ is not defined, and it is replaced with $\sigma^{\hat{p}}$.}
	\begin{align*}
		\sigma(\ALG) &\geq\sigma^{\hat{p}}(\ALG)-\varepsilon k/2\geq\kappa\cdot \sigma^{\hat{p}}(\OPT)-\varepsilon k/2 \\
		&\geq\kappa\cdot (\sigma(\OPT)-\varepsilon k/2)-\varepsilon k/2 \\
		&=\kappa\cdot\sigma(\OPT)-(1+\kappa)\varepsilon k/2\geq(\kappa-\varepsilon)\sigma(\OPT). 
	\end{align*}
	%	The first and third inequalities are due to Lemma \ref{lemma: estimate influence}.
	The second inequality holds since $\ALG$ is a $\kappa$-approximation on $\hat{G}$.
	The last inequality holds since $\sigma(\OPT)\geq k\geq(1+\kappa)k/2$.
\end{proof}

Compared with our learning algorithms for network inference, Algorithm \ref{algorithm: influence maximization under edge probabilities estimation assumption} has an additional overhead of $n^6/k^2$ in the number of cascades.
This is because it needs to estimate edge probabilities within an additive error of at most $\varepsilon k/(2n^3)$.
One can also invoke known network inference algorithms other than ours in Algorithm \ref{algorithm: influence maximization under edge probabilities estimation assumption} to obtain a similar approximate solution, but as discussed above, this only incurs higher sample complexity.
We are not aware of any approach to reduce the sample complexity and leave it as an interesting open problem.

\subsubsection{IMS under Assumptions Independent of the Network}
\label{subsection: influence maximization under assumptions independent of the network}

Condition 1 of Assumption \ref{assumption: estimate edge probabilities} depends on the diffusion network, and hence our 
Algorithm \ref{algorithm: influence maximization under edge probabilities estimation assumption} may not be applicable to
all networks.
%As discussed earlier, Assumption \ref{assumption: estimate edge probabilities} is restricted as the parameter $1/\alpha$ is very large for some graphs, rendering the previous algorithms to use a large number of cascades.
In this section, we show that under an alternative assumption (Assumption \ref{assumption: influence maximization under assumption independent of graph}), which is entirely independent of the diffusion network, 
there still exists a \emph{constant} approximation IMS 	algorithm (Algorithm \ref{algorithm: influence maximization unde assumption independent of graph}).

\begin{assumption}[IMS under the IC model, independent of the network]
	\label{assumption: influence maximization under assumption independent of graph}
	For some constant $c>0$ and parameter $\gamma\in(0,1/2]$,
	\begin{enumerate}
		\item $\sum_{u\in V} q_u\leq ck$.
		\item $\gamma\leq q_u\leq 1-\gamma$ for all $u\in V$.
	\end{enumerate}
\end{assumption}

Assumption \ref{assumption: influence maximization under assumption independent of graph} consists of two conditions.
The condition $\sum_{u\in V} q_u \leq ck$ replaces the condition $\ap(v)\leq 1-\alpha$ in Assumption \ref{assumption: estimate edge probabilities}.
It means that a random seed set drawn from the product distribution $\mathcal{D}$ has an expected size at most linear in $k$ (but not necessarily bounded above by $k$).
Assumption \ref{assumption: influence maximization under assumption independent of graph} puts forward two plausible requirements for the seed distribution $\mathcal{D}$ and has no requirement for the underlying network.
Thus, in principle, one can handle \emph{any} social networks, as long as the seed set sampling is reasonable
according to Assumption \ref{assumption: influence maximization under assumption independent of graph}.

We now describe the high-level idea of Algorithm \ref{algorithm: influence maximization unde assumption independent of graph}.
It might be surprising at first glance that one can remove the condition $\ap(v)\leq 1-\alpha$ for all $v\in V$.
After all, it is very hard to learn information about incoming edges of $v$ if $\ap(v)$ is very close to $1$.
To handle this difficulty, recall that $\ap(v)$ is defined as the activation probability of $v$ in one time step.
Hence, if $\ap(v)$ is close to $1$, $v$ will be active with high probability starting from a random seed set.
The observation suggests that one can divide nodes into two parts according to their $\ap(\cdot)$.
For the nodes with small $\ap(\cdot)$, Assumption \ref{assumption: estimate edge probabilities} is satisfied and one can find a good approximation for them via a similar approach as Algorithm \ref{algorithm: influence maximization under edge probabilities estimation assumption}.
For the nodes with large $\ap(\cdot)$, a random seed set is already a good approximation for them.
So there is no need to learn their incoming edges.
A technical issue here is that a random seed set  may not be a feasible solution for the maximization task.
This is why we introduce Condition 1 of Assumption \ref{assumption: influence maximization under assumption independent of graph}, by which the expected size of the seed set is at most linear in $k$.
With the condition, we can replace the random seed set by its random subset of size $k$ while keeping a constant approximation.
To summarize, we find two candidate solutions whose union must be a good approximation over the whole network.
If we choose one of them randomly, we will finally obtain a feasible solution with constant approximation.

Following the guidance of the above idea, Algorithm \ref{algorithm: influence maximization unde assumption independent of graph} first computes an estimate $\widehat{\ap}(v)$ of $\ap(v)$ for all $v\in V$ and partitions $V$ into two disjoint subsets $V_1=\{v\in V\mid \widehat{\ap}(v)<1-\delta/(4n)\}$ and $V_2=V\setminus V_1$.
It then estimates the probabilities of incoming edges of $V_1$ using Algorithm \ref{algorithm: estimate edge probabilities} and sets the probabilities of incoming edges of $V_2$ to $1$ directly for technical reasons.
The constructed graph is denoted by $\hat{G}$.
Let $T_1$ be a $\kappa$-approximation on $\hat{G}$ and $T_2=S_{1,0}$ be the first random seed set.
Finally, Algorithm \ref{algorithm: influence maximization unde assumption independent of graph} 
selects $T_1$ or $T_2$ with equal probability,  and if it selects $T_2$ while $|T_2| > k$, it further 
selects a random subset of $T_2$ with size $k$ as the final output seed set $\ALG$.

\begin{algorithm}[tb]
	\caption{IMS-IC under Assumption \ref{assumption: influence maximization under assumption independent of graph}}
	\label{algorithm: influence maximization unde assumption independent of graph}
	\begin{algorithmic}[1]
		\REQUIRE A set of cascades $(S_{i,0},S_{i,1},\cdots,S_{i,n-1})_{i=1}^t$, $k\in\mathbb{N}_+$, 
		error probability $\delta>0$, number of samples $t'\in [t]$ used to estimate ${\ap}(v)$'s.
		\STATE Set $V_1=V$ and $V_2=\emptyset$.
		\FOR{each $v\in V$}
		\STATE Use the first $t'$ samples $(S_{i,0},S_{i,1},\cdots,S_{i,n-1})_{i=1}^{t'}$
		to compute $\widehat{\ap}(v)=t^v/t'$, where $t^v=|\{i\in[t']\mid v\in S_{i,1}\}|$.
		\IF{$\widehat{\ap}(v)\geq 1-\delta/(4n)$}
		\STATE Set $\hat{p}_{uv}=1$ for all $u\in V$.
		\STATE $V_1=V_1\setminus\{v\}$ and $V_2=V_2\cup\{v\}$.
		\ENDIF
		\ENDFOR
		\STATE 	$\{\hat{p}_{uv}\}_{u\in V,v\in V_1}=\mbox{Estimate-Edge-Probabilities}$\\$((S_{i,0},S_{i,1},\cdots,S_{i,n-1})_{i=t'+1}^t)$ on $V_1$.
		\COMMENT{With accuracy $\varepsilon k/(2n^3)$, $\alpha=\delta/(6n)$ in Assumption \ref{assumption: estimate edge probabilities}.}
		\STATE Let $\hat{G}=(V,E,\hat{p})$.
		\STATE $T_1=A(\hat{G},k)$, where $A$ is a $\kappa$-approximation IM algorithm.
		\STATE $T_2=S_{1,0}$.
		\STATE Let $T$ be a random set by picking $T_1$ and $T_2$ with equal probability. If $T\leq k$, let $\ALG=T$; otherwise, let $\ALG$ be a uniformly random subset of $T$ with $|\ALG|=k$.
		\STATE \textbf{return} $\ALG$.
	\end{algorithmic}
\end{algorithm}

%\begin{algorithm}[t]
%	\SetAlgoLined
%	\KwIn{A set of independent cascades $(S_{i,0},S_{i,1},\cdots,S_{i,n-1})_{i=1}^t$, $k\in\mathbb{N}$, error probability $\delta>0$.}
%	%		\wei{It looks like the algorithm
%	%			uses $\delta$, so I put it here.}
%	Set $V_1=V$.
%	
%	\ForEach{$v\in V$}{
%		Compute $\widehat{\\ap}(v)=t^v/t$, where $t^v=|\{i\in[t]\mid v\in S_{i,1}\}|$.
%		
%		\If{$\widehat{\\ap}(v)\geq 1-\delta/(4n)$}{
%			Set $\hat{p}_{uv}=1$ for all $u\in V$.
%			
%			$V_1=V_1\setminus\{v\}$.
%		}
%	}
%	
%	$\{\hat{p}_{uv}\}_{u\in V,v\in V_1}=\mbox{Estimate-Edge-Probability}((S_{i,0},S_{i,1},\cdots,S_{i,n-1})_{i=1}^t)$ on $V_1$, with estimation accuracy being $\varepsilon k/(2n^3)$, $\alpha=\delta/(6n)$ in Assumption \ref{assumption: estimate edge probabilities}.
%	%    \wei{Again, check if we need parameter $\varepsilon k/(2n^3)$. Also saying $\alpha=\delta/(6n)$ is not needed for the algorithm, right?
%	%    	I think it is just for the analysis.}
%	
%	Let $\hat{G}=(V,E,\hat{p})$.
%	
%	$T_1=A(\hat{G},k)$, where $A$ is a $\kappa$-approximation algorithm for influence maximization.
%	
%	$T_2=S_{1,0}$.
%	
%	Let $T$ be a random set by picking $T_1$ and $T_2$ with equal probability. If $T\leq k$, let $\ALG=T$; otherwise, let $\ALG$ be a uniformly random subset of $T$ with $|\ALG|=k$.
%	
%	\Return $\ALG$.
%	\caption{Influence Maximization from Samples under Assumption \ref{assumption: influence maximization under assumption independent of graph}}
%%	\label{algorithm: influence maximization unde assumption independent of graph}
%\end{algorithm}

We now give an analysis of Algorithm \ref{algorithm: influence maximization unde assumption independent of graph}.
Our analysis requires a technical lemma (Lemma \ref{lemma: R will cover a lot}).
Informally, Lemma \ref{lemma: R will cover a lot} means that for any set $R\subseteq V$, the influence of the seed set $S$ when setting the probabilities of all incoming edges of $R$ to $1$ is no larger than the influence of $S\cup R$.

\begin{lemma}
	\label{lemma: R will cover a lot}
	Let $G=(V,E,p)$ be a directed graph and $R\subseteq V$.
	Let $G'=(V,E,p')$ be a directed graph obtained from $G$ as follows: $p'_{uv}=1$ if $v\in R$ and $p'_{uv}=p_{uv}$ otherwise.
	Then, for any $S\subseteq V$, we have $\sigma^p(S\cup R)\geq \sigma^{p'}(S)$.
\end{lemma}

\begin{proof}
	To prove Lemma \ref{lemma: R will cover a lot}, we will use \emph{live-edge graphs} to interpret the IC model and help understand the influence spread.
	Formally, a live-edge graph corresponding to the IC model is a random subgraph of $G$ such that each edge $(u,v)$ is picked independently with probability $p_{uv}$.
	Let $\sigma_u(S)$ be the probability that $u$ is reachable from $S$ in the live-edge graph.
	Then, $\sigma_u(S)$ is also the probability that $u$ is activated by $S$ and hence the influence spread function $\sigma(S)=\sum_{u\in V} \sigma_u(S)$.
	For a node $u\in V$, a fixed edge set $A$ and a seed set $S$, let $\textbf{1}_u(A, S)$ be the indicator variable that $u$ is reachable from $S$ through edges in $A$.
	Then, $\sigma_u(S)$ can be written as
	\[ \sigma_u(S)=\sum_{A\subseteq E} \prod_{(a,b)\in A}p_{ab}\prod_{(a,b)\in E\setminus A}(1-p_{ab})\textbf{1}_u(A, S). \]
	
	Let $B=\{(u,v)\mid u\in V, v\in R\}$ be the set of incoming edges of $R$.
	By definition, $p'_{ab}=1$ for $(a,b)\in B$ and $p'_{ab}=p_{ab}$ otherwise.
	Let $\mathcal{B}\subseteq B$ be a random subset of $B$ such that for each edge $(a,b)\in B$, $(a,b)\in\mathcal{B}$ independently with probability $p_{ab}$.
	For a fixed $u\in V$, we consider its activation probability.
	First, we have
	\begin{align*}
		\sigma_u^p(S\cup R) &=\sum_{A\subseteq E\setminus B} \prod_{(a,b)\in A}p_{ab}\prod_{(a,b)\in E\setminus(A\cup B)}(1-p_{ab})\E[\textbf{1}_u(A\cup \mathcal{B}, S\cup R)] \\
		&=\sum_{A\subseteq E\setminus B} \prod_{(a,b)\in A}p_{ab}\prod_{(a,b)\in E\setminus(A\cup B)}(1-p_{ab})\textbf{1}_u(A\cup B, S\cup R).
	\end{align*}
	The expectation in the first equality is taken over the randomness of $\mathcal{B}$.
	The second equality holds since $R$ itself is part of the seed set, and it makes no difference whether its incoming edges $B$ are picked into the live-edge graph.
	Next, by the definition of $p'$,
	\[ \sigma_u^{p'}(S)=\sum_{A\subseteq E\setminus B} \prod_{(a,b)\in A}p_{ab}\prod_{(a,b)\in E\setminus(A\cup B)}(1-p_{ab})\textbf{1}_u(A\cup B, S). \]
	Clearly, $\textbf{1}_u(A\cup B, S\cup R)\geq \textbf{1}_u(A\cup B, S)$.
	It follows that $\sigma^p_u(S\cup R)\geq\sigma^{p'}_u(S)$ and therefore $\sigma^p(S\cup R)\geq\sigma^{p'}(S)$.
\end{proof}

%\begin{restatable}{lemma}{lemRWillCoverALot}
%	\label{lemma: R will cover a lot}
%	Let $G=(V,E,p)$ be a directed graph and $R\subseteq V$.
%	Let $G'=(V,E,p')$ be a directed graph obtained from $G$ as follows: $p'_{uv}=1$ if $v\in R$ and $p'_{uv}=p_{uv}$ otherwise.
%	Then, for any $S\subseteq V$, we have $\sigma^p(S\cup R)\geq \sigma^{p'}(S)$.
%\end{restatable}

\begin{theorem}
	Under Assumption \ref{assumption: influence maximization under assumption independent of graph}, suppose that the number of cascades $t\geq\frac{36864}{\varepsilon^2\delta^2\gamma^3}\frac{n^8}{k^2}\ln\frac{36n}{\delta}+\frac{72n^2}{\delta^2}\ln\frac{12n}{\delta}$,
	and the number of samples used to estimate ${\ap}(v)$'s is $t'=\frac{72n^2}{\delta^2}\ln\frac{12n}{\delta}$.
	Let $A$ be a $\kappa$-approximation algorithm for influence maximization.
	Assume that $k\geq\frac{3}{c}\ln\frac{3}{\delta}$.
	Let $\ALG$ be the set returned by Algorithm \ref{algorithm: influence maximization unde assumption independent of graph} and $\OPT$ be the optimal solution on the original graph.
	We have that $\ALG$ is a feasible solution, and
	\[ \Pr[\E[\sigma(\ALG)]\geq\min\left\{\frac{1}{2c},1\right\}\frac{\kappa-\varepsilon}{2}\sigma(\OPT)]\geq1-\delta, \]
	where the probability is taken over the randomness of $(S_{i,0},S_{i,1},\cdots,S_{i,n-1})_{i=1}^t$ and the expectation is taken over the randomness from line 13 of Algorithm \ref{algorithm: influence maximization unde assumption independent of graph}.
	%\wei{
	%I think the above formula is not accurate, since $\ALG$ is a random set, but we need to take expectation of this random set,
	%	not leave its randomness to the overall probability.
	%	\[ \Pr[\E[\sigma(\ALG)]\geq\min\left\{\frac{1}{2c},1\right\}\frac{\kappa-\varepsilon}{2}\sigma(\OPT)]\geq1-\delta. \]
	%	and we need to be clear that the expectation is taken over the randomness of $\ALG$ selected either from $T_1$ or $T_2$ 
	%	(we do not care about the randomness of selecting $k$ nodes from $T_2$!?), and the probability is taken from
	%	the randomness of seed sampling and IC propagation, right?
	%Previous theorems should be made clear too.
	%}
	
	%\wei{By the way, I checked our ICML20 paper, and noticed that the theorem there explicitly writes the distribution
	%	under $\Pr$, which is better, but it writes as $\Pr_{S_1, \ldots, S_t \sim \mathcal{D}}$, but more accurately it 
	%	should be $\Pr_{(S_1, \ldots, S_t) \sim \mathcal{D}^t}$, right?
	%	We can fix it perhaps, and at least I can update my homepage version, or we can update arxiv.
	%}

	%	If $t\geq\frac{2048|L|^2|R|^4}{\varepsilon^2\delta^2\gamma^3}\ln\frac{32|L||R|}{\delta}$, then $\E[f^p(\ALG)]\geq\frac{1}{2}(1-e^{-1}-O(\varepsilon))f^p(\OPT)$.
\end{theorem}

\begin{proof}
	Let $G=(V,E,p)$ be the original graph.
	Let $V_1=\{v\in V\mid \widehat{\ap}(v)\leq1-\frac{\delta}{4n}\}$ and $V_2=V\setminus V_1$, defined as in Algorithm \ref{algorithm: influence maximization unde assumption independent of graph}.
	Let $B=\{(u,v)\mid u\in V, v\in V_2\}$ be the set of all edges pointing to some node in $V_2$.
	Let $G'=(V,E,p')$ be a directed graph obtained from $G$ as follows: $p'_{uv}=1$ if $(u,v)\in B$ and $p'_{uv}=p_{uv}$ otherwise.
	Let $\hat{G}=(V,E,\hat{p})$ be a directed graph obtained from $G'$ by replacing $p_{uv}$ with $\hat{p}_{uv}$ for any $(u,v)\notin B$.
	Clearly, $\hat{G}$ is exactly the same graph we constructed in Algorithm \ref{algorithm: influence maximization unde assumption independent of graph}.
	
	For any node $v\in V$, by Lemma \ref{lemma: additive chernoff bound}, when $t'=\frac{72n^2}{\delta^2}\ln\frac{12n}{\delta}$,
	\begin{align*}
		&\Pr[|\widehat{\ap}(v)-\ap(v)|\geq \delta/(12n)] \\
		&\leq \exp(-t'(\delta/(12n))\min(1/5,(\delta/(12n))\cdot (1/\ap(v))) \\
		&+\exp(-t'(\delta/(12n))^2/2ap(v))\\
		&\leq \delta/(6n)+\delta/(6n)=\delta/(3n).
	\end{align*}
	By union bound, with probability $1-\delta/3$, for all nodes $v\in V$, $|\widehat{\ap}(v)-\ap(v)|\leq \delta/(12n)$.
	Specifically, for a node $v\in V_1$ with $\widehat{\ap}(v)\leq 1-\delta/(4n)$, we have $\ap(v)\leq \widehat{\ap}(v)+\delta/(12n)\leq 1-\delta/(6n)$.
	For a node $v\in V_2$ with $\widehat{\ap}(v)>1-\delta/(4n)$, we have $\ap(v)\geq \widehat{\ap}(v)-\delta/(12n)\geq 1-\delta/(3n)$.
	
	%	Since for any $v\in V_2$, $\\ap(v)\geq 1-\delta/(3n)$, it means that $\Pr[v\notin S_{1,1}]\leq\delta/(3n)$.
	%	By union bound, $\Pr[\exists v\in V_2, v\notin S_{1,1}]\leq\delta/3$.
	%	Thus, $\Pr[V_2\subseteq S_{1,1}]\geq 1-\delta/3$.
	%	Assume that $V_2\subseteq S_{1,1}$.
	%	Then, $\sigma^p(T_2)=\sigma^p(S_{1,0})\geq \sigma^p(V_2)$.
	%	\wei{There could be a slight gap between assuming $V_2\subseteq S_{1,1}$ and concluding
	%		$\sigma^p(S_{1,0})\geq \sigma^p(V_2)$, because $S_{1,1}$ is just one realization of the first step with $S_{1,0}$ as the seed set.
	%		I guess what we really want to assume is that given $S_{1,0}$ as a seed set, in one time step all nodes in $V_2$ are activated. This
	%		happens with probability $1-\delta/3$.
	%		Also, it is a bit unclear what randomness is included in this probability --- it should include the randomness of IC propagation, the randomness of $V_2$,
	%		but we treat $S_{1,0}$ as fixed?
	%		The following is an alternative explanation.
	%	}
	Since for any $v\in V_2$, $\ap(v)\geq 1-\delta/(3n)$, by union bound, it means that with probability at least $1-\delta/3$ in one time step all nodes in $V_2$
	are activated.
	Assume that indeed all nodes in $V_2$ are activated in one time step.
	Then, we have $\sigma^p(T_1\cup T_2)=\sigma^p(T_1\cup S_{1,0})\geq \sigma^p(T_1\cup V_2)$.
	By plugging $S=T_1,R=V_2$ into Lemma \ref{lemma: R will cover a lot}, we obtain $\sigma^p(T_1\cup V_2)\geq \sigma^{p'}(T_1)$.
	Therefore, by submodularity of $\sigma$, $\sigma^p(T_1)+\sigma^p(T_2)\geq \sigma^p(T_1\cup T_2)\geq \sigma^{p'}(T_1)$.
	
	Since $\{\hat{p}\}_{u,v}$ is obtained by running Estimate-Edge-Probability with parameters $\varepsilon k/(2n^3), \alpha=\delta/(6n), \gamma$, when the number of cascades $t-t'\geq\frac{36864}{\varepsilon^2\delta^2\gamma^3}\frac{n^8}{k^2}\ln\frac{36n}{\delta}$, with probability $1-\delta/3$, we have $|\hat{p}_{uv}-p_{uv}|\leq\varepsilon k/(2n^3)$ for any $(u,v)\notin B$.
	Since $\hat{p}_{uv}=p'_{uv}=1$ for $(u,v)\in B$, we have $\|\hat{p}-p'\|_1\leq\varepsilon k/(2n)$.
	Therefore, by Lemma \ref{lemma: estimate influence}, for any $S\subseteq V$, $|\sigma^{\hat{p}}(S)-\sigma^{p'}(S)|\leq\varepsilon k/2$.
	We have
	\begin{align*}
		\sigma^{p'}(T_1) &\geq \sigma^{\hat{p}}(T_1)-\varepsilon k/2 \geq \kappa\cdot \sigma^{\hat{p}}(\OPT)-\varepsilon k/2 \\
		&\geq\kappa\cdot(\sigma^{p'}(\OPT)-\varepsilon k/2)-\varepsilon k/2 \\
		&\geq\kappa\cdot(\sigma^p(\OPT)-\varepsilon k/2)-\varepsilon k/2\geq(\kappa-\varepsilon)\sigma^p(\OPT).
	\end{align*}
	The second inequality holds since $T_1$ is a $\kappa$ approximation of $\OPT$ on $\hat{G}$.
	The forth inequality holds since $\sigma^{p'}(S)\geq\sigma^p(S)$ for any $S\subseteq V$, due to $p'_{uv}\geq p_{uv}$ for any $(u,v)\in E$.
	The last inequality holds as long as $(1+\kappa)k/2\leq k\leq \sigma^p(\OPT)$, which holds trivially since $\kappa\leq 1$ and $\sigma^p(\OPT)\geq k$.
	
	Combining the previous inequalities, we have
	\[ \sigma^p(T_1)+\sigma^p(T_2)\geq \sigma^{p'}(T_1)\geq(\kappa-\varepsilon)\sigma^p(\OPT), \]
	which implies that $\E[\sigma^p(T)]=\frac{1}{2}(\sigma^p(T_1)+\sigma^p(T_2))\geq \frac{1}{2}(\kappa-\varepsilon)\sigma^p(\OPT)$.
	
	Finally, since $\sum_{u\in V} q_u\leq ck$, $\Pr[|S_{1,0}|\geq 2ck]\leq e^{-ck/3}\leq\delta/3$ when $k\geq\frac{3}{c}\ln\frac{3}{\delta}$.
	Assume that $|T_2|=|S_{1,0}|\leq 2ck$.
	If $T=T_1$ or $T=T_2$ but $|T_2|\leq k$, then $\ALG=T$.
	If $T=T_2$ and $|T_2|>k$, then $\ALG$ is a uniform subset of $T$ with size $k$.
	Since $\sigma(\cdot)$ is submodular, we have
	$\E[\sigma^p(\ALG)]\geq\min\left\{\frac{1}{2c},1\right\}\E[\sigma^p(T)]\geq\min\left\{\frac{1}{2c},1\right\}\frac{\kappa-\varepsilon}{2}\sigma^p(S^*)$.
	
	To conclude, by union bound, with probability at least $1-\delta$, $\ALG$ is a feasible solution and $\E[\sigma^p(\ALG)]\geq\min\left\{\frac{1}{2c},1\right\}\frac{\kappa-\varepsilon}{2}\sigma^p(\OPT)$.
\end{proof}

\paragraph{Improving the approximation ratio}
Compared with Algorithm \ref{algorithm: influence maximization under edge probabilities estimation assumption}, Algorithm \ref{algorithm: influence maximization unde assumption independent of graph} has a worse (though still constant) approximation ratio.
We show that if the constant $c$ in Assumption \ref{assumption: influence maximization under assumption independent of graph} equals to some prescribed small $\varepsilon\in(0,1/3)$, we can modify Algorithm \ref{algorithm: influence maximization unde assumption independent of graph} to be an approximation-preserving algorithm as follows: let $T_1=A(\hat{G},(1-2\varepsilon)k)$ and returns $T_1\cup T_2$ directly.
It is easy to see that the modified algorithm works since $T_1$ loses little in the approximation ratio and $T_1\cup T_2$ is feasible with high probability.
The formal procedure is presented in Algorithm \ref{algorithm: influence maximization unde assumption independent of graph, strong approx ratio} and its guarantee is presented below.

%\begin{assumption}
%	\label{assumption: influence maximization under assumption independent of graph, strong approx ratio}
%	We require several assumptions. \wei{Need to make it clear if the $\varepsilon$ below is the same $\varepsilon$ as in the approximation ratio $\kappa - \varepsilon$.}
%	\begin{enumerate}
%		\item $\sum_{u\in V} q_u\leq \varepsilon k$.
%		\item $\gamma\leq q_u\leq 1-\gamma$ for all $u\in V$.
%	\end{enumerate}
%\end{assumption}

%\wei{But for Algorithm \ref{algorithm: influence maximization unde assumption independent of graph, strong approx ratio}, there is a slight chance that it returns a seed set of size
%	greater than $k$. Should we explicitly handle this? Either in the algorithm give up the output, or say something in the theorem, 
%	e.g. $\Pr[\sigma(\ALG)\geq(\kappa-\varepsilon)\sigma(\OPT) \wedge |\ALG| \le k]\geq1-\delta$ ?}

\begin{algorithm}[tb]
	\caption{IMS-IC under Assumption \ref{assumption: influence maximization under assumption independent of graph} with $c=\varepsilon$}
	\label{algorithm: influence maximization unde assumption independent of graph, strong approx ratio}
	\begin{algorithmic}[1]
		\REQUIRE A set of cascades $(S_{i,0},S_{i,1},\cdots,S_{i,n-1})_{i=1}^t$ and $k\in\mathbb{N}_+$, parameter $\varepsilon\in(0,1/3)$, error probability $\delta>0$, number of samples $t'\in [t]$ used to estimate ${\ap}(v)$'s.
		\STATE Set $V_1=V$ and $V_2=\emptyset$.
		\FOR{each $v\in V$}
		\STATE Use the first $t'$ samples $(S_{i,0},S_{i,1},\cdots,S_{i,n-1})_{i=1}^{t'}$
		to compute $\widehat{\ap}(v)=t^v/t'$, where $t^v=|\{i\in[t']\mid v\in S_{i,1}\}|$.
		\IF{$\widehat{\ap}(v)\geq 1-\delta/(4n)$}
		\STATE Set $\hat{p}_{uv}=1$ for all $u\in V$.	
		\STATE $V_1=V_1\setminus\{v\}$ and $V_2=V_2\cup\{v\}$.
		\ENDIF
		\ENDFOR
		
		\STATE $\{\hat{p}_{uv}\}_{u\in V,v\in V_1}=\mbox{Estimate-Edge-Probabilities}$\\$((S_{i,0},S_{i,1},\cdots,S_{i,n-1})_{i=t'+1}^t)$ on $V_1$.
		\COMMENT{With accuracy $\varepsilon k/(2n^3)$, $\alpha=\delta/(6n)$ in Assumption \ref{assumption: estimate edge probabilities}.}
		\STATE Let $\hat{G}=(V,E,\hat{p})$.
		\STATE $T_1=A(\hat{G},(1-2\varepsilon)k)$, where $A$ is a $\kappa$-approximation algorithm for influence maximization.
		\STATE $T_2=S_{1,0}$.
		\STATE \textbf{return} $\ALG=T_1\cup T_2$.
	\end{algorithmic}
\end{algorithm}

\begin{theorem}
	Under Assumption \ref{assumption: influence maximization under assumption independent of graph} with $c=\varepsilon\in(0,1/3)$, suppose that the number of cascades $t\geq\frac{36864}{\varepsilon^2\delta^2\gamma^3}\frac{n^8}{k^2}\ln\frac{36n}{\delta}+\frac{72n^2}{\delta^2}\ln\frac{12n}{\delta}$
	and the number of samples used to estimate ${\ap}(v)$'s is $t'=\frac{72n^2}{\delta^2}\ln\frac{12n}{\delta}$.
	Let $A$ be an $\kappa$-approximation algorithm for influence maximization.
	Assume that $k\geq\frac{3}{\varepsilon}\ln\frac{3}{\delta}$.
	Let $\ALG$ be the set returned by Algorithm \ref{algorithm: influence maximization unde assumption independent of graph, strong approx ratio} and $\OPT$ be the optimal solution on the original graph.
	We have
	\[ \Pr[|\ALG| \le k \wedge \sigma(\ALG)\geq(\kappa-3\varepsilon)\sigma(\OPT)]\geq1-\delta. \]	
\end{theorem}

\begin{proof}
	By a similar analysis for Algorithm \ref{algorithm: influence maximization unde assumption independent of graph}, $\sigma^p(T_1\cup T_2)\geq\sigma^{p'}(T_1)$ with probability at least $1-\delta/3$, and
	for any $S\subseteq V$, $|\sigma^{\hat{p}}(S)-\sigma^{p'}(S)|\leq\varepsilon k/2$ with probability at least $1-\delta/3$.
	We thus have
	\begin{align*}
		\sigma^{p'}(T_1) &\geq \sigma^{\hat{p}}(T_1)-\varepsilon k/2 \\
		&\geq \kappa(1-2\varepsilon)\cdot \sigma^{\hat{p}}(\OPT)-\varepsilon k/2 \\
		&\geq\kappa(1-2\varepsilon)\cdot(\sigma^{p'}(\OPT)-\varepsilon k/2)-\varepsilon k/2 \\
		&\geq\kappa(1-2\varepsilon)\cdot(\sigma^p(\OPT)-\varepsilon k/2)-\varepsilon k/2  \\
		&\geq(\kappa-2\varepsilon)\cdot\sigma^p(\OPT)-(1+\kappa)\varepsilon k/2 \\
		&\geq(\kappa-3\varepsilon)\cdot\sigma^p(\OPT).
	\end{align*}
	The second inequality holds since $T_1$ is a $\kappa(1-2\varepsilon)$ approximation of $\OPT$ on $\hat{G}$.
	The forth inequality holds since $\sigma^{p'}(S)\geq\sigma^p(S)$ for any $S\subseteq V$, due to $p'_{uv}\geq p_{uv}$ for any $(u,v)\in E$.
	The last inequality holds as long as $(1+\kappa)k/2\leq k\leq \sigma^p(\OPT)$, which holds trivially since $\kappa\leq 1$ and $\sigma^p(\OPT)\geq k$.
	Therefore, we have $\sigma^p(\ALG)=\sigma^p(T_1\cup T_2)\geq(\kappa-3\varepsilon)\sigma^p(\OPT)$.
	
	Finally, since $\sum_{u\in V} q_u\leq \varepsilon k$, by Lemma \ref{lemma: multiplicative chernoff bound}, $\Pr[|S_{1,0}|\geq 2\varepsilon k]\leq e^{-\varepsilon k/3}\leq\delta/3$ when $k\geq\frac{3}{\varepsilon}\ln\frac{3}{\delta}$.
	If $|S_{1,0}|\leq 2\varepsilon k$, $|\ALG|=|T_1\cup T_2|\leq (1-2\varepsilon)k+2\varepsilon k=k$.
	
	To conclude, by union bound, with probability at least $1-\delta$, $|\ALG|\leq k$ and $\sigma(\ALG)\geq(\kappa-3\varepsilon)\sigma(\OPT)$.
\end{proof}

\section{Algorithms under the LT Model}
\label{section: the linear threshold model}

In this section, we solve both network inference and IMS problems under the LT model.
In Section \ref{section: network inference under the LT model}, we present a network inference algorithm.
In Section \ref{section: IMS under the LT model}, we present an IMS algorithm.

\subsection{Network Inference}
\label{section: network inference under the LT model}

In this section, we present a network inference algorithm under the LT model which finds an estimate $\hat{w}$ of $w$ such that $|\hat{w}_{uv}-w_{uv}|\leq \varepsilon$ for all $u,v\in V$.
The algorithm is similar to the one under the IC model and based on
the key observation on the connection between $w_{uv}$ and the one-step activation probabilities
$\ap(v)$ and $\ap(v\cond\bar{u})$.
We first give an explicit expression of $\ap(v)$ under the LT model.

\begin{lemma}
	\label{lemma: closed form of ap(v)}
	Under the LT model, $\ap(v)=q_v+(1-q_v)\sum_{u\in N(v)}q_uw_{uv}$ for any $v\in V$.
\end{lemma}

\begin{proof}
	Fix $v\in V$. For $v$ to be active in one time step, $v$ is either picked as a seed or activated by its in-neighbors which are picked as seeds.
	By the fact that $\mathcal{D}$ is a product distribution, in our notations, we have
	\[ \ap(v)=q_v+(1-q_v)\ap(v\cond \bar{v}). \]
	Let $S_0\sim\mathcal{D}$ and $R=S_0\cap N(v)$.
	For $u\in N(v)$, let $X_u\in\{0,1\}$ indicate whether $u\in S_0$.
	Then, $\E[X_u]=q_u$.
	Let $X=\sum_{u\in N(v)}w_{uv}X_u$.
	By the linearity of expectation, $\E[X]=\sum_{u\in N(v)}q_uw_{uv}$.
	
	On the other hand, by law of total probability, the facts that $\mathcal{D}$ is a product distribution and $\mathcal{D},\theta_v$ are independent,
	\[ \ap(v\cond \bar{v})=\sum_{R\subseteq N(v)}\Pr_{\mathcal{D}}[R]\Pr_{\theta_v}[v\in S_1\mid v\notin S_0,R]=\sum_{R\subseteq N(v)}\Pr_{\mathcal{D}}[R]\sum_{u\in R}w_{uv}=\E[X]. \]
	Therefore, we have $\ap(v\cond \bar{v})=\sum_{u\in N(v)}q_uw_{uv}$.
	
	By plugging it back, we finally obtain $\ap(v)=q_v+(1-q_v)\sum_{u\in N(v)}q_uw_{uv}$.
\end{proof}

Next, we derive a closed-form expression for the edge weight $w_{uv}$ from Lemma \ref{lemma: closed form of ap(v)}.

\begin{lemma}
	\label{lemma: closed form of edge weight}
	Under the LT model, for any $u,v\in V$ with $u\not=v$,
	\[ w_{uv}=\frac{\ap(v)-\ap(v\cond\bar{u})}{q_u(1-q_v)}. \]
\end{lemma}

\begin{proof}
	To avoid confusion, we write the underlying graph $G$ and the seed distribution $\mathcal{D}$ explicitly in notation $\ap(\cdot)$, namely $\ap(v)=\ap_{G,\mathcal{D}}(v)$.
	Consider the subgraph $G'=G \setminus \{u\}$ by removing node $u$.
	Since when considering one-step activation of $v$, node $u$ not being the seed is equivalent to removing it from the graph, we have
	\[ \ap_{G,\mathcal{D}}(v\cond\bar{u})=\ap_{G',\mathcal{D}}(v). \]
	Next, by Lemma \ref{lemma: closed form of ap(v)}, we have
	\begin{align*}
		\ap_{G',\mathcal{D}}(v) &=q_v+(1-q_v)\sum_{u'\in N(v)\setminus\{u\}}q_{u'}w_{u'v}. \\
		\ap_{G,\mathcal{D}}(v) &=q_v+(1-q_v)\sum_{u'\in N(v)}q_{u'}w_{u'v}.
	\end{align*}
	We therefore obtain that
	\[ \ap_{G,\mathcal{D}}(v)-\ap_{G,\mathcal{D}}(v\cond\bar{u})=(1-q_v)q_uw_{uv}. \]
	By rearranging the equality, we obtain
	\[ w_{uv}=\frac{\ap_{G,\mathcal{D}}(v)-\ap_{G,\mathcal{D}}(v\cond\bar{u})}{q_u(1-q_v)}. \]
\end{proof}

Equipped with the lemma, we are able to estimate $w_{uv}$ by estimating $q_u,\ap(v)$ and $\ap(v\cond\bar{u})$ respectively from cascade samples.
Let $t_u=|\{i\in[t]\mid u\in S_{i,0} \}|$ be the number of cascades where $u$ is a seed, $t_{\bar{u}}=|\{i\in[t]\mid u\notin S_{i,0} \}|$ the number of cascades where $u$ is not a seed, $t^v=|\{i\in[t]\mid v\in S_{i,1}\}|$ the number of cascades where $v$ is activated in one time step and $t_{\bar{u}}^v=|\{i\in[t]\mid u\notin S_{i,0},v\in S_{i,1}\}|$ the number of cascades where $u$ is not a seed and $v$ is activated in one time step.
Then, $\hat{q}_u=t_u/t, \widehat{\ap}(v)=t^v/t$ and $\widehat{\ap}(v\cond \bar{u})=t_{\bar{u}}^v/t_{\bar{u}}$ are good estimates of $q_u,\ap(v)$ and $\ap(v\cond\bar{u})$, respectively.
The formal procedure is formulated as Algorithm \ref{algorithm: estimate edge weights}.

Algorithm \ref{algorithm: estimate edge weights} needs to work under Assumption \ref{assumption: estimate edge weights} below, which means that the probability of a node $u\in V$ being selected as a seed is neither too low nor too high.
This assumption ensures that the above quantities are well estimated.
Compared with Assumption \ref{assumption: estimate edge probabilities}, Assumption \ref{assumption: estimate edge weights} does not need the condition $\ap(v)\leq 1-\alpha$ and hence imposes no requirement on the network.
This is because the condition $\ap(v)\leq 1-\alpha$ gives a lower bound for $1-\ap(v\cond \bar{u})$ and leads to a tighter estimate of it, while in the closed-form expression of $w_{uv}$, $1-\ap(v\cond \bar{u})$ does not appear in the denominator, and hence a loose estimate in the lack of the condition still suffices.

\begin{algorithm}[tb]
	\caption{Estimate Edge Weights}
	\label{algorithm: estimate edge weights}
	\begin{algorithmic}[1]
		\REQUIRE A set of cascades $(S_{i,0},S_{i,1},\cdots,S_{i,n-1})_{i=1}^t$.
		\ENSURE $\{\hat{w}_{uv}\}_{u,v\in V}$ such that $|\hat{w}_{uv}-w_{uv}|\leq\varepsilon$ for all $u,v\in V$.
		\FOR{each $u\in V$}
		\STATE Compute $\hat{q}_u=t_u/t$, where $t_u=|\{i\in[t]\mid u\in S_{i,0} \}|$.
		\ENDFOR
		\FOR{each $v\in V$}
		\STATE Compute $\widehat{\ap}(v)=t^v/t$, where $t^v=|\{i\in[t]\mid v\in S_{i,1}\}|$.
		\ENDFOR
		\FOR{each $v\in V$}
		\FOR{each $u\in V$}
		\STATE Compute $\widehat{\ap}(v\cond \bar{u})=t_{\bar{u}}^v/t_{\bar{u}}$, where $t_{\bar{u}}=|\{i\in[t]\mid u\notin S_{i,0} \}|$ and $t_{\bar{u}}^v=|\{i\in[t]\mid u\notin S_{i,0},v\in S_{i,1}\}|$.
		\STATE Let $\hat{w}_{uv}=\frac{\widehat{\ap}(v)-\widehat{\ap}(v\cond\bar{u})}{\hat{q}_u(1-\hat{q}_v)}$.
		\ENDFOR
		\ENDFOR
		\STATE \textbf{return} $\{\hat{w}_{uv}\}_{u,v\in V}$.
	\end{algorithmic}
\end{algorithm}

\begin{assumption}[Edge weights estimation under the LT model]
	\label{assumption: estimate edge weights}
	%	The following are the assumptions for estimating edge probabilities.
	For some parameter $\gamma\in(0,1/2]$,
	\[ \gamma\leq q_u\leq 1-\gamma, \forall u\in V. \]
\end{assumption}

We now give an analysis of Algorithm \ref{algorithm: estimate edge weights}.
Lemma \ref{lemma: estimate ap() in learning edge weights} gives the number of samples we need to estimate $q_u,\ap(v)$ and $\ap(v\cond\bar{u})$ within a small accuracy.
Its proof is exactly the same as that of Lemma \ref{lemma: estimate ap() in learning edge probabilities} and therefore omitted.

\begin{lemma}
	\label{lemma: estimate ap() in learning edge weights}
	Under Assumption \ref{assumption: estimate edge weights}, for any $\eta\in(0,4/5),\delta\in(0,1)$, for $\hat{q}_u$, $\widehat{\ap}(v)$, and $\widehat{\ap}(v\cond\bar{u})$ defined in Algorithm~\ref{algorithm: estimate edge weights},
	if the number of samples $t\geq\frac{16}{\gamma\eta^2}\ln\frac{12n}{\delta}$, with probability at least $1-\delta$, we have
	\begin{enumerate}
		%		\item $\forall i, t_{u_i}-tq_i\leq \gamma t/2$.
		\item $|\hat{q}_u-q_u|\leq\eta q_u$ for all $u\in V$,
		\item $|\widehat{\ap}(v)-\ap(v)|\leq \eta$ for all $v\in V$,
		\item $|\widehat{\ap}(v\cond\bar{u})-\ap(v\cond\bar{u})|\leq\eta$ for all $u,v\in V$.
	\end{enumerate}
\end{lemma}

\begin{theorem}
	\label{theorem: estimate edge weights}
	Under Assumption \ref{assumption: estimate edge weights}, for any $\varepsilon,\delta\in(0,1)$, let $\eta=\varepsilon\gamma^2/4<1/8$.
	Let $\{\hat{w}_{uv}\}_{u,v\in V}$ be the set of edge weights returned by Algorithm \ref{algorithm: estimate edge weights}.
	If the number of cascades $t\geq \frac{16}{\gamma\eta^2}\ln\frac{12n}{\delta}=\frac{256}{\varepsilon^2\gamma^6}\ln\frac{12n}{\delta}$, with probability at least $1-\delta$, for any $u,v\in V$, $|\hat{w}_{uv}-w_{uv}|\leq\varepsilon$.
\end{theorem}

\begin{proof}
	With probability at least $1-\delta$, all the events in Lemma \ref{lemma: estimate ap() in learning edge weights} occur.
	We assume that this is exactly the case in the following.
	By the choice of $\eta$ and the assumption that $\gamma \leq q_u\leq 1-\gamma$, we have
	\begin{equation}  \label{eq:eta-LT}
		\eta\leq\frac{\varepsilon\gamma}{4}(1-q_v)\leq\frac{\varepsilon}{4}q_u(1-q_v). 
	\end{equation}
	To prove $\hat{w}_{uv}\leq w_{uv}+\varepsilon$, we have
	\begin{align*}
		\hat{w}_{uv} &=\frac{\widehat{\ap}(v)-\widehat{\ap}(v\cond u)}{\hat{q}_u(1-\hat{q}_v)} \\
		&\leq \frac{\ap(v)-\ap(v\cond u)+2\eta}{(1-\eta)q_u(1-q_v-\eta q_v)} \\
		&\leq \frac{\ap(v)-\ap(v\cond u)+2\eta}{(1-\eta)(1-\varepsilon\gamma/4)q_u(1-q_v)} \\
		&\leq \frac{w_{uv}+\varepsilon/2}{(1-\eta)(1-\varepsilon\gamma/4)} \\
		&\leq w_{uv}+\varepsilon.
	\end{align*}
	The first inequality holds due to Lemma \ref{lemma: estimate ap() in learning edge weights}.
	The second inequality holds by applying the first inequality in Eq.~\eqref{eq:eta-LT} and the fact that $q_v\leq 1$.
	The third inequality holds due to Lemma \ref{lemma: closed form of edge weight} and the second inequality in Eq.~\eqref{eq:eta-LT}.
	The correctness of the last inequality follows the same argument as Theorem \ref{theorem: estimate edge probabilities}.
	
	On the other hand, to prove $\hat{w}_{uv}\geq w_{uv}-\varepsilon$, first assume that $w_{uv}\geq\varepsilon$, since otherwise the claim would be trivial for $\hat{w}_{uv}\geq0$.
	We now have
	\begin{align*}
		\hat{w}_{uv} &=\frac{\widehat{\ap}(v)-\widehat{\ap}(v\cond u)}{\hat{q}_u(1-\hat{q}_v)} \\
		&\geq \frac{\ap(v)-\ap(v\cond u)-2\eta}{(1+\eta)q_u(1-q_v+\eta q_v)} \\
		&\geq \frac{\ap(v)-\ap(v\cond u)-2\eta}{(1+\eta)(1+\varepsilon\gamma/4)q_u(1-q_v)} \\
		&\geq \frac{w_{uv}-\varepsilon/2}{(1+\eta)(1+\varepsilon\gamma/4)} \\
		&\geq w_{uv}-\varepsilon.
	\end{align*}
	The first inequality holds due to Lemma \ref{lemma: estimate ap() in learning edge weights}.
	The second inequality holds by applying the first inequality in Eq.~\eqref{eq:eta-LT} and the fact that $q_v\leq 1$.
	The third inequality holds due to Lemma \ref{lemma: closed form of edge weight} and the second inequality in Eq.~\eqref{eq:eta-LT}.
	The correctness of the last inequality follows the same argument as Theorem \ref{theorem: estimate edge probabilities}.
\end{proof}

As in the IC model, Algorithm \ref{algorithm: estimate edge weights} can be adapted to recover the graph structure.
For compactness, we omit the adapted algorithm.

In general, the $\hat{w}$ returned by Algorithm \ref{algorithm: estimate edge weights} does not necessarily satisfies the \emph{normalization condition} that $\sum_{u\in N(v)}\hat{w}_{uv}\leq 1$ for all $v\in V$.
The condition is crucial in defining the \emph{live-edge} graph under the LT model, which helps the design of fast IM algorithm such as RR-set \citep{BorgsBCL14} and our IMS algorithm in the next subsection.
For this reason, we achieve the normalization condition by rescaling $\hat{w}$, as Corollary \ref{corollary: estimate edge weights} below shows.

\begin{corollary}
	\label{corollary: estimate edge weights}
	For any $\varepsilon,\delta\in(0,1)$, let $\{\hat{w}_{uv}\}_{u,v\in V}$ be the edge weights returned by Algorithm \ref{algorithm: estimate edge weights} under Assumption \ref{assumption: estimate edge weights} with $t\geq \frac{1024}{\varepsilon^2\gamma^6}D^2\ln\frac{12n}{\delta}$ cascade samples, where $D$ is the maximum in-degree of the underlying graph $G$.
	Let $w'_{uv}=\frac{\hat{w}_{uv}}{1+\varepsilon/2}$.
	Then, with probability at least $1-\delta$,
	\begin{enumerate}
		\item $\sum_{u\in N(v)}w'_{uv}\leq 1$ for all $v\in V$, and 
		\item $|w'_{uv}-w_{uv}|\leq\varepsilon$ for any $u,v\in V$.
	\end{enumerate}
\end{corollary}

\begin{proof}
	By Theorem \ref{theorem: estimate edge weights},  with probability at least $1-\delta$, for any $u,v\in V$, $|\hat{w}_{uv}-w_{uv}|\leq\varepsilon/(2D)$.
	We assume that this is exactly the case in the following.
	
	Fix $v\in V$, we have
	\[ \sum_{u\in N(v)}w'_{uv}=\sum_{u\in N(v)}\frac{\hat{w}_{uv}}{1+\varepsilon/2}\leq\frac{1}{1+\varepsilon/2}\sum_{u\in N(v)}(w_{uv}+\varepsilon/(2D))\leq\frac{1}{1+\varepsilon/2}(1+\varepsilon/2)=1. \]
	The last inequality holds since the original $w$ satisfies $\sum_{u\in N(v)} w_{uv}\leq 1$.
	
	Next, we have
	\begin{align*}
		|w'_{uv}-w_{uv}| &=\frac{1}{1+\varepsilon/2}|\hat{w}_{uv}-(1+\varepsilon/2)w_{uv}|\leq \frac{1}{1+\varepsilon/2}|\hat{w}_{uv}-w_{uv}|+\frac{\varepsilon/2}{1+\varepsilon/2}|w_{uv}| \\
		&\leq\frac{\varepsilon/(2D)}{1+\varepsilon/2}+\frac{\varepsilon/2}{1+\varepsilon/2}\leq\frac{\varepsilon}{2}(1+\frac{1}{D})\leq\varepsilon.
	\end{align*}
\end{proof}

\subsection{Influence Maximization from Samples}
\label{section: IMS under the LT model}

In this section, we present an IMS algorithm (Algorithm \ref{algorithm: IMS-LT}) under the LT model.
Our algorithm is \emph{approximation-preserving} and imposes no requirement on the network.
It follows the canonical learning-and-then-optimization approach by first learning a surrogate graph $G'=(V,E,w')$ from the cascades and then executing any $\kappa$-approximation algorithm $A$ for standard influence maximization on $G'$ to obtain a solution as output.
The construction of $G'$ builds on Algorithm \ref{algorithm: estimate edge weights} and is obtained by first estimating all the edge weights to a sufficiently small additive error and then rescale them as in Corollary \ref{corollary: estimate edge weights} to meet the normalization condition.
Algorithm \ref{algorithm: IMS-LT} works under Assumption \ref{assumption: estimate edge weights}, since Algorithm \ref{algorithm: estimate edge weights} does.
Consequently, Algorithm \ref{algorithm: IMS-LT} can handle arbitrary social networks, since Assumption \ref{assumption: estimate edge weights} imposes no requirement for the network.

As Lemma \ref{lemma: estimate influence} for the IC model, to prove the correctness of Algorithm \ref{algorithm: IMS-LT}, we need to bound the difference between two LT influence functions with different edge parameters by the difference of the parameters.
We show this in Lemma \ref{lemma: estimate LT influence} below.
Note that to apply Lemma \ref{lemma: estimate LT influence}, the normalization condition must hold.
This explains why Algorithm \ref{algorithm: IMS-LT} rescales the estimated edge weights before it runs a standard IM algorithm.

\begin{algorithm}[tb]
	\caption{IMS-LT under Assumption \ref{assumption: estimate edge weights}}
	\label{algorithm: IMS-LT}
	\begin{algorithmic}[1]
		\REQUIRE A set of cascades $(S_{i,0},S_{i,1},\cdots,S_{i,n-1})_{i=1}^t$ and $k\in\mathbb{N}_+$.
		\STATE $\{\hat{w}_{uv}\}_{u,v\in V}=\mbox{Estimate-Edge-Weights}((S_{i,0},S_{i,1},\cdots,S_{i,n-1})_{i=1}^t)$.
		\COMMENT{ With estimation accuracy $\varepsilon k/(2Dn^3)$.}
		\STATE Let $w'_{uv}=\frac{\hat{w}_{uv}}{1+\varepsilon/2}$ for all $u,v\in V$.
		\STATE Let $G'=(V,E,w')$.
		\STATE Let $\ALG=A(G',k)$, where $A$ is a $\kappa$-approximation IM algorithm.
		\STATE \textbf{return} $\ALG$.
	\end{algorithmic}
\end{algorithm}

We first present Lemma \ref{lemma: estimate LT influence} and its proof.

\begin{lemma}
	\label{lemma: estimate LT influence}
	Under the LT model, for any two edge weight vectors $w,w'$ such that (1) $\|w-w'\|_1\leq\varepsilon/n$, and (2) both $w$ and $w'$ satisfy the normalization condition, we have $|\sigma^w(S)-\sigma^{w'}(S)|\leq\varepsilon$ for all $S$.
\end{lemma}

\begin{proof}
	To prove Lemma \ref{lemma: estimate LT influence}, we will use \emph{live-edge graphs} to interpret the LT model and help understand the influence spread.
	For node $v\in V$, its (final) activation probability is denoted by $\sigma_v^w$.
	Clearly, $\sigma^w(S)=\sum_{v\in V}\sigma_v^w(S)$ for any seed set $S\subseteq V$.
	For node $b\in V$, let $E_b$ be the set of its incoming edges.
	The live-edge graph under the LT model is generated as follows: For each node $b\in V$, among all of its incoming edges, $(u,v)\in E_b$ is selected exclusively as the single live edge with probability $w_{uv}$, and no edge is selected with probability $1-\sum_{a\in N(b)}w_{ab}$.
	The selection is independent among all nodes $b\in V$.
	Let $A\subseteq E$ be the edge set of some realization of the live-edge graph.
	Then, $A$ satisfies that $|A\cap E_b|\leq 1$ for all $b\in V$.
	For convenience, let $\mathcal{A}=\{A\subseteq E\mid |A\cap E_b|\leq 1,\forall b\in V\}$ and $A\cap E_b=\{e(b)\}$ if $A\cap E_b\neq\emptyset$.
	Finally, let $\textbf{1}_v(A,S)$ be the indicator variable such that $\textbf{1}_v(A,S)=1$ if and only if $v$ is reachable from $S$ via edges in $A$.
	it is proved that $\sigma_v^w(S)$ can be written as the summation of $\textbf{1}_v(A,S)$ over all realizations of the live-edge graph \citep[see][]{ChenLC13}:
	\[ \sigma_v^w(S)=\sum_{A\in \mathcal{A}}\prod_{b\in V:A\cap E_b\neq\emptyset}w_{e(b)}\prod_{b\in V:A\cap E_b=\emptyset}(1-\sum_{a\in N(b)}w_{ab})\textbf{1}_v(A,S). \]
	
	With the above interpretation, we now bound the $L_{\infty}$ norm of the gradient of $\sigma_v^w$.
	\begin{align*}
		&\left|\frac{\partial\sigma_v^w(S)}{\partial w_{cd}}\right| \\
		&=\Bigg|\frac{\partial}{\partial w_{cd}}\Bigg[w_{cd}\sum_{A\in\mathcal{A},A\cap E_d=\emptyset}\prod_{b\neq d:A\cap E_b\neq\emptyset}w_{e(b)}\prod_{b\neq d:A\cap E_b=\emptyset}(1-\sum_{a\in N(b)}w_{ab})\textbf{1}_v(A\cup\{(c,d)\},S) \\
		&+(1-\sum_{a\in N(d)}w_{ad})\sum_{A\in\mathcal{A},A\cap E_d=\emptyset}\prod_{b\neq d:A\cap E_b\neq\emptyset}w_{e(b)}\prod_{b\neq d:A\cap E_b=\emptyset}(1-\sum_{a\in N(b)}w_{ab})\textbf{1}_v(A,S)\Bigg]\Bigg| \\
		&=\Bigg|\sum_{A\in\mathcal{A},A\cap E_d=\emptyset}\prod_{b\neq d:A\cap E_b\neq\emptyset}w_{e(b)}\prod_{b\neq d:A\cap E_b=\emptyset}(1-\sum_{a\in N(b)}w_{ab})\textbf{1}_v(A\cup\{(c,d)\},S) \\
		&-\sum_{A\in\mathcal{A},A\cap E_d=\emptyset}\prod_{b\neq d:A\cap E_b\neq\emptyset}w_{e(b)}\prod_{b\neq d:A\cap E_b=\emptyset}(1-\sum_{a\in N(b)}w_{ab})\textbf{1}_v(A,S)\Bigg| \\
		&\leq\sum_{A\in\mathcal{A},A\cap E_d=\emptyset}\prod_{b\neq d:A\cap E_b\neq\emptyset}w_{e(b)}\prod_{b\neq d:A\cap E_b=\emptyset}(1-\sum_{a\in N(b)}w_{ab}) \\
		&=1-\sum_{a\in N(d)}w_{ad}\leq 1.
	\end{align*}
%    \begin{align*}
%    	&\left|\frac{\partial\sigma_v^w(S)}{\partial w_{cd}}\right| \\
%    	&=\Bigg|\frac{\partial}{\partial w_{cd}}\Bigg[w_{cd}\sum_{A\subseteq E:|A\cap E_b|\leq 1,\forall b\in V,A\cap E_d=\emptyset}\prod_{b\neq d:A\cap E_b=\{e(b)\}}w_{e(b)}\prod_{b\neq d:A\cap E_b=\emptyset}(1-\sum_{a\in N(b)}w_{ab})\textbf{1}_v(A\cup\{(c,d)\},S) \\
%    	&+(1-\sum_{a\in N(d),a\neq c}w_{ad}-w_{cd})\sum_{A\subseteq E:|A\cap E_b|\leq 1,\forall b\in V,A\cap E_d=\emptyset}\prod_{b\neq d:A\cap E_b=\{e(b)\}}w_{e(b)}\prod_{b\neq d:A\cap E_b=\emptyset}(1-\sum_{a\in N(b)}w_{ab})\textbf{1}_v(A,S)\Bigg]\Bigg| \\
%    	&=\Bigg|\sum_{A\subseteq E:|A\cap E_b|\leq 1,\forall b\in V,A\cap E_d=\emptyset}\prod_{b\neq d:A\cap E_b=\{e(b)\}}w_{e(b)}\prod_{b\neq d:A\cap E_b=\emptyset}(1-\sum_{a\in N(b)}w_{ab})\textbf{1}_v(A\cup\{(c,d)\},S) \\
%    	&-\sum_{A\subseteq E:|A\cap E_b|\leq 1,\forall b\in V,A\cap E_d=\emptyset}\prod_{b\neq d:A\cap E_b=\{e(b)\}}w_{e(b)}\prod_{b\neq d:A\cap E_b=\emptyset}(1-\sum_{a\in N(b)}w_{ab})\textbf{1}_v(A,S)\Bigg| \\
%    	&\leq\Bigg|\sum_{A\subseteq E:|A\cap E_b|\leq 1,\forall b\in V,A\cap E_d=\emptyset}\prod_{b\neq d:A\cap E_b=\{e(b)\}}w_{e(b)}\prod_{b\neq d:A\cap E_b=\emptyset}(1-\sum_{a\in N(b)}w_{ab})\Bigg| \\
%    	&=1-\sum_{a\in N(d)}w_{ad}\leq 1.
%    \end{align*}
	The first equality holds since when computing the partial derivative at $w_{cd}$, we only need to concern the terms where $w_{cd}$ appears, and $w_{cd}$ appears when (1) $(c,d)$ is selected as the single live edge of $d$, or (2) no incoming edges of $d$ are selected.
	The inequality follows from the fact hat $0\leq |\textbf{1}_v(A\cup\{(c,d)\}S)-\textbf{1}_v(A,S)|\leq 1$.
	The last equality holds since the summation is over $A\cap E_d=\emptyset$, which equals to the probability that no incoming edges of $d$ are selected.
	Clearly, the above inequality means that $\|\nabla_w \sigma_v^w(S)\|_{\infty}\leq 1$.
	
	By the mean-value theorem, there is a $\bar{w}=sw+(1-s)w'$ for some $s\in[0,1]$ such that
	\[ |\sigma_v^w(S)-\sigma_v^{w'}(S)|= \|\nabla_{\bar{w}} \sigma_v^w(S)\|_{\infty}\|w-w'\|_1\leq\varepsilon/n. \]
	Therefore, $|\sigma^w(S)-\sigma^{w'}(S)|\leq\sum_{v\in V}|\sigma_v^w(S)-\sigma_v^{w'}(S)|\leq\varepsilon$, which completes the proof.
\end{proof}

%\noindent\textbf{Remark.} In the above lemma, both $w$ and $w'$ need to satisfy $\sum_{u\in N(v)} w_{uv}\leq 1$ for all $v\in V$ and $\sum_{u\in N(v)} w'_{uv}\leq 1$ for all $v\in V$, respectively, since only in this case, the live edge graph is well-defined.

Finally, the performance of Algorithm \ref{algorithm: IMS-LT} is presented in the following theorem.

\begin{theorem}
	%\label{theorem: influence maximization under edge probabilities estimation assumption}
	Under Assumption \ref{assumption: estimate edge weights},
	for any $\varepsilon\in (0,1)$ and $k \in\mathbb{N}_+$, 
	suppose that the number of cascades $t\geq\frac{4096}{\varepsilon^2\gamma^3}\frac{D^2n^6}{k^2}\ln\frac{12n}{\delta}$.
	Let $A$ be a $\kappa$-approximation algorithm for influence maximization.
	Let $\ALG$ be the set returned by Algorithm \ref{algorithm: IMS-LT} and $\OPT$ be the optimal solution on the original graph. Under Assumption \ref{assumption: estimate edge weights}, we have
	\[ \Pr[\sigma(\ALG)\geq(\kappa-\varepsilon)\sigma(\OPT)]\geq1-\delta. \]
\end{theorem}

\begin{proof}
	By Corollary \ref{corollary: estimate edge weights}, with probability at least $1-\delta$, we have 1) $\sum_{u\in N(v)}w'_{uv}\leq 1$ for all $v\in V$, and 2) for any $u,v\in V$, $|w'_{uv}-w_{uv}|\leq\varepsilon k/(2n^3)$.
	Hence, $\|w-w'\|_1=\sum_{u,v\in V}|w_{uv}-w'_{uv}|\leq\varepsilon k/(2n)$.
	Applying this condition to Lemma~\ref{lemma: estimate LT influence}, we have that $|\sigma^w(S)-\sigma^{w'}(S)|\leq \varepsilon k/2$ for every seed set $S$.
	We thus have
	%	\wei{fixed some notation below. First, $\sigma_v$ should be replaced with $\sigma$; second, $\hat{\sigma}$ is not defined, and it is replaced with $\sigma^{\hat{p}}$.}
	\begin{align*}
		\sigma(\ALG) &\geq\sigma^{\hat{p}}(\ALG)-\varepsilon k/2\geq\kappa\cdot \sigma^{\hat{p}}(\OPT)-\varepsilon k/2 \\
		&\geq\kappa\cdot (\sigma(\OPT)-\varepsilon k/2)-\varepsilon k/2 \\
		&=\kappa\cdot\sigma(\OPT)-(1+\kappa)\varepsilon k/2\geq(\kappa-\varepsilon)\sigma(\OPT). 
	\end{align*}
	%	The first and third inequalities are due to Lemma \ref{lemma: estimate influence}.
	The second inequality holds since $\ALG$ is a $\kappa$-approximation on $\hat{G}$.
	The last inequality holds since $\sigma(\OPT)\geq k\geq(1+\kappa)k/2$.
\end{proof}

\section{Conclusion and Future Work}
\label{section: conclusion}

In this paper, we conduct a rigorous theoretical treatment to the influence maximization from samples (IMS) problem under both IC and LT models, and provide several end-to-end IMS algorithms with
constant approximation guarantee.
We also provide novel and efficient algorithms for network inference with weaker assumptions.

There are many future directions to extend and improve this work.
First, our IMS algorithms require a large number of samples (though polynomial) since we have to estimate edge probabilities to a very high accuracy.
It is very interesting to investigate how to improve the sample complexity by leveraging sparsity and different importance of edges in the networks.
Second, our samples contain activation sets at every step. 
One can further study how to do IMS when we only observe the final activation set.
Other directions include studying IMS for other stochastic diffusion models (for example, the cumulative activation model in \citep{ShanCLSZ19}), relaxing the independent seed node sampling assumption, 
and going beyond influence maximization to study other optimization tasks directly from data samples.


%\newpage

\appendix
\section{Comparing Assumptions}
%    \section{Comparing Assumptions}
We summarize the assumptions used in \citep{NetrapalliS12,NarasimhanPS15} below and show that they are strictly stronger than our assumptions.

\begin{assumption}[Assumptions in \citealt{NetrapalliS12}]
	For some parameters $\alpha,\beta\in(0,1)$,
	\begin{enumerate}
		\item $p_{uv}\geq\beta$ for all $(u,v)\in E$.
		\item (Correlation decay) $\sum_{u\in N^{in}(v)}p_{uv}<1-\alpha$ for all $v\in V$.
		\item $q_u d_v<1/2$ for all $u,v\in V$.
	\end{enumerate}
\end{assumption}

\begin{assumption}[Assumptions in \citealt{NarasimhanPS15}]
	For some parameters $\beta\geq\alpha\in(0,1/2)$ and $\gamma\in(0,1)$,
	\begin{enumerate}
		\item $p_{uv}\geq\beta$ for all $(u,v)\in E$.
		\item $1-\prod_{u\in N^{in}(v)}(1-p_{uv})\leq 1-\alpha$ for all $v\in V$.
		\item $\gamma\leq q_u\leq 1-\gamma$ for all $u\in V$.
	\end{enumerate}
\end{assumption}

\begin{lemma}
	$\ap(v)\leq 1-\prod_{u\in N^{in}(v)}(1-p_{uv})\leq \sum_{u\in N^{in}(v)} p_{uv}$.
\end{lemma}

\begin{proof}
	The first inequality follows from $\ap(v)=1-\prod_{u\in N^{in}(v)}(1-q_up_{uv})\leq 1-\prod_{u\in N^{in}(v)}(1-p_{uv})$,
	since $q_u\leq 1$ for all $u\in V$.
	The second inequality follows from the claim below.
	\begin{claim}
		For any $x_1,\cdots,x_n\in[0,1]$, $\prod_{i=1}^{n}(1-x_i)\geq 1-\sum_{i=1}^{n}x_i$.
	\end{claim}
	{\em Proof of the Claim.\ }
	The claim holds trivially when $n=1$.
	Assume that the claim holds for any $k<n$.
	Then,
	\[ \prod_{i=1}^{n}(1-x_i)=\prod_{i=1}^{n-1}(1-x_i)(1-x_n)\geq(1-\sum_{i=1}^{n-1}x_i)(1-x_n)\geq 1-\sum_{i=1}^{n}x_i. \]
	The two inequalities both hold by induction.
\end{proof}

% \appendix

% \section{Sample Appendix Section}
% \label{sec:sample:appendix}
% Lorem ipsum dolor sit amet, consectetur adipiscing elit, sed do eiusmod tempor section \ref{sec:sample1} incididunt ut labore et dolore magna aliqua. Ut enim ad minim veniam, quis nostrud exercitation ullamco laboris nisi ut aliquip ex ea commodo consequat. Duis aute irure dolor in reprehenderit in voluptate velit esse cillum dolore eu fugiat nulla pariatur. Excepteur sint occaecat cupidatat non proident, sunt in culpa qui officia deserunt mollit anim id est laborum.

 \bibliographystyle{elsarticle-harv} 
 \bibliography{OPSS-journal}

\begin{thebibliography}{28}
\expandafter\ifx\csname natexlab\endcsname\relax\def\natexlab#1{#1}\fi
\providecommand{\url}[1]{\texttt{#1}}
\providecommand{\href}[2]{#2}
\providecommand{\path}[1]{#1}
\providecommand{\DOIprefix}{doi:}
\providecommand{\ArXivprefix}{arXiv:}
\providecommand{\URLprefix}{URL: }
\providecommand{\Pubmedprefix}{pmid:}
\providecommand{\doi}[1]{\href{http://dx.doi.org/#1}{\path{#1}}}
\providecommand{\Pubmed}[1]{\href{pmid:#1}{\path{#1}}}
\providecommand{\bibinfo}[2]{#2}
\ifx\xfnm\relax \def\xfnm[#1]{\unskip,\space#1}\fi
%Type = Inproceedings
\bibitem[{Abrahao et~al.(2013)Abrahao, Chierichetti, Kleinberg and
  Panconesi}]{AbrahaoCKP13}
\bibinfo{author}{Abrahao, B.D.}, \bibinfo{author}{Chierichetti, F.},
  \bibinfo{author}{Kleinberg, R.}, \bibinfo{author}{Panconesi, A.},
  \bibinfo{year}{2013}.
\newblock \bibinfo{title}{Trace complexity of network inference}, in:
  \bibinfo{booktitle}{Proceedings of the 19th {ACM} {SIGKDD} International
  Conference on Knowledge Discovery and Data Mining},
  \bibinfo{address}{Chicago, IL, USA}. pp. \bibinfo{pages}{491--499}.
%Type = Book
\bibitem[{Alon and Spencer(2008)}]{AlonS08}
\bibinfo{author}{Alon, N.}, \bibinfo{author}{Spencer, J.H.},
  \bibinfo{year}{2008}.
\newblock \bibinfo{title}{The Probabilistic Method, Third Edition}.
\newblock Wiley-Interscience series in discrete mathematics and optimization,
  \bibinfo{publisher}{Wiley}.
%Type = Inproceedings
\bibitem[{Badanidiyuru et~al.(2012)Badanidiyuru, Dobzinski, Fu, Kleinberg,
  Nisan and Roughgarden}]{BadanidiyuruDFKNR12}
\bibinfo{author}{Badanidiyuru, A.}, \bibinfo{author}{Dobzinski, S.},
  \bibinfo{author}{Fu, H.}, \bibinfo{author}{Kleinberg, R.},
  \bibinfo{author}{Nisan, N.}, \bibinfo{author}{Roughgarden, T.},
  \bibinfo{year}{2012}.
\newblock \bibinfo{title}{Sketching valuation functions}, in:
  \bibinfo{booktitle}{Proceedings of the 23th Annual {ACM-SIAM} Symposium on
  Discrete Algorithms}, \bibinfo{address}{Kyoto, Japan}. pp.
  \bibinfo{pages}{1025--1035}.
%Type = Inproceedings
\bibitem[{Balkanski et~al.(2017a)Balkanski, Immorlica and
  Singer}]{BalkanskiIS17}
\bibinfo{author}{Balkanski, E.}, \bibinfo{author}{Immorlica, N.},
  \bibinfo{author}{Singer, Y.}, \bibinfo{year}{2017}a.
\newblock \bibinfo{title}{The importance of communities for learning to
  influence}, in: \bibinfo{booktitle}{Advances in Neural Information Processing
  Systems 30 ({NIPS} 2017)}, \bibinfo{address}{Long Beach, CA, {USA}}. pp.
  \bibinfo{pages}{5862--5871}.
%Type = Inproceedings
\bibitem[{Balkanski et~al.(2016)Balkanski, Rubinstein and
  Singer}]{BalkanskiRS16}
\bibinfo{author}{Balkanski, E.}, \bibinfo{author}{Rubinstein, A.},
  \bibinfo{author}{Singer, Y.}, \bibinfo{year}{2016}.
\newblock \bibinfo{title}{The power of optimization from samples}, in:
  \bibinfo{booktitle}{Advances in Neural Information Processing Systems 29
  ({NIPS} 2016)}, \bibinfo{address}{Barcelona, Spain}. pp.
  \bibinfo{pages}{4017--4025}.
%Type = Inproceedings
\bibitem[{Balkanski et~al.(2017b)Balkanski, Rubinstein and
  Singer}]{BalkanskiRS17}
\bibinfo{author}{Balkanski, E.}, \bibinfo{author}{Rubinstein, A.},
  \bibinfo{author}{Singer, Y.}, \bibinfo{year}{2017}b.
\newblock \bibinfo{title}{The limitations of optimization from samples}, in:
  \bibinfo{booktitle}{Proceedings of the 49th Annual {ACM} {SIGACT} Symposium
  on Theory of Computing}, \bibinfo{address}{Montreal, QC, Canada}. pp.
  \bibinfo{pages}{1016--1027}.
%Type = Inproceedings
\bibitem[{Borgs et~al.(2014)Borgs, Brautbar, Chayes and Lucier}]{BorgsBCL14}
\bibinfo{author}{Borgs, C.}, \bibinfo{author}{Brautbar, M.},
  \bibinfo{author}{Chayes, J.T.}, \bibinfo{author}{Lucier, B.},
  \bibinfo{year}{2014}.
\newblock \bibinfo{title}{Maximizing social influence in nearly optimal time},
  in: \bibinfo{booktitle}{Proceedings of the 25th Annual {ACM-SIAM} Symposium
  on Discrete Algorithms}, \bibinfo{address}{Portland, Oregon, USA}. pp.
  \bibinfo{pages}{946--957}.
%Type = Book
\bibitem[{Chen et~al.(2013)Chen, Lakshmanan and Castillo}]{ChenLC13}
\bibinfo{author}{Chen, W.}, \bibinfo{author}{Lakshmanan, L.V.S.},
  \bibinfo{author}{Castillo, C.}, \bibinfo{year}{2013}.
\newblock \bibinfo{title}{Information and Influence Propagation in Social
  Networks}.
\newblock Synthesis Lectures on Data Management, \bibinfo{publisher}{Morgan
  {\&} Claypool Publishers}.
%Type = Inproceedings
\bibitem[{Chen et~al.(2020)Chen, Sun, Zhang and Zhang}]{ChenSZZ20}
\bibinfo{author}{Chen, W.}, \bibinfo{author}{Sun, X.}, \bibinfo{author}{Zhang,
  J.}, \bibinfo{author}{Zhang, Z.}, \bibinfo{year}{2020}.
\newblock \bibinfo{title}{Optimization from structured samples for coverage
  functions}, in: \bibinfo{booktitle}{Proceedings of the 37th International
  Conference on Machine Learning}, \bibinfo{address}{Virtual Event}. pp.
  \bibinfo{pages}{1715--1724}.
%Type = Article
\bibitem[{Chen et~al.(2016)Chen, Wang, Yuan and Wang}]{CWYW16}
\bibinfo{author}{Chen, W.}, \bibinfo{author}{Wang, Y.}, \bibinfo{author}{Yuan,
  Y.}, \bibinfo{author}{Wang, Q.}, \bibinfo{year}{2016}.
\newblock \bibinfo{title}{Combinatorial multi-armed bandit and its extension to
  probabilistically triggered arms}.
\newblock \bibinfo{journal}{J. Mach. Learn. Res.} \bibinfo{volume}{17},
  \bibinfo{pages}{1--33}.
\newblock \bibinfo{note}{A preliminary version appeared as Chen, Wang, and
  Yuan, ``Combinatorial multi-armed bandit: general framework, results and
  applications'', ICML 2013.}
%Type = Inproceedings
\bibitem[{Daneshmand et~al.(2014)Daneshmand, Gomez{-}Rodriguez, Song and
  Sch{\"{o}}lkopf}]{DaneshmandGSS14}
\bibinfo{author}{Daneshmand, H.}, \bibinfo{author}{Gomez{-}Rodriguez, M.},
  \bibinfo{author}{Song, L.}, \bibinfo{author}{Sch{\"{o}}lkopf, B.},
  \bibinfo{year}{2014}.
\newblock \bibinfo{title}{Estimating diffusion network structures: Recovery
  conditions, sample complexity {\&} soft-thresholding algorithm}, in:
  \bibinfo{booktitle}{Proceedings of the 31th International Conference on
  Machine Learning}, \bibinfo{address}{Beijing, China}. pp.
  \bibinfo{pages}{793--801}.
%Type = Inproceedings
\bibitem[{Du et~al.(2014)Du, Liang, Balcan and Song}]{DuLBS14}
\bibinfo{author}{Du, N.}, \bibinfo{author}{Liang, Y.}, \bibinfo{author}{Balcan,
  M.}, \bibinfo{author}{Song, L.}, \bibinfo{year}{2014}.
\newblock \bibinfo{title}{Influence function learning in information diffusion
  networks}, in: \bibinfo{booktitle}{Proceedings of the 31th International
  Conference on Machine Learning}, \bibinfo{address}{Beijing, China}. pp.
  \bibinfo{pages}{2016--2024}.
%Type = Inproceedings
\bibitem[{Du et~al.(2013)Du, Song, Gomez{-}Rodriguez and Zha}]{DuSGZ13}
\bibinfo{author}{Du, N.}, \bibinfo{author}{Song, L.},
  \bibinfo{author}{Gomez{-}Rodriguez, M.}, \bibinfo{author}{Zha, H.},
  \bibinfo{year}{2013}.
\newblock \bibinfo{title}{Scalable influence estimation in continuous-time
  diffusion networks}, in: \bibinfo{booktitle}{Advances in Neural Information
  Processing Systems 26 ({NIPS} 2013)}, \bibinfo{address}{Lake Tahoe, Nevada,
  USA}. pp. \bibinfo{pages}{3147--3155}.
%Type = Inproceedings
\bibitem[{Du et~al.(2012)Du, Song, Smola and Yuan}]{DuSSY12}
\bibinfo{author}{Du, N.}, \bibinfo{author}{Song, L.}, \bibinfo{author}{Smola,
  A.J.}, \bibinfo{author}{Yuan, M.}, \bibinfo{year}{2012}.
\newblock \bibinfo{title}{Learning networks of heterogeneous influence}, in:
  \bibinfo{booktitle}{Advances in Neural Information Processing Systems 25
  ({NIPS} 2012)}, \bibinfo{address}{Lake Tahoe, Nevada, USA}. pp.
  \bibinfo{pages}{2789--2797}.
%Type = Article
\bibitem[{Feige(1998)}]{Feige98}
\bibinfo{author}{Feige, U.}, \bibinfo{year}{1998}.
\newblock \bibinfo{title}{A threshold of ln \emph{n} for approximating set
  cover}.
\newblock \bibinfo{journal}{J. {ACM}} \bibinfo{volume}{45},
  \bibinfo{pages}{634--652}.
%Type = Inproceedings
\bibitem[{Gomez{-}Rodriguez et~al.(2011)Gomez{-}Rodriguez, Balduzzi and
  Sch{\"{o}}lkopf}]{Gomez-RodriguezBS11}
\bibinfo{author}{Gomez{-}Rodriguez, M.}, \bibinfo{author}{Balduzzi, D.},
  \bibinfo{author}{Sch{\"{o}}lkopf, B.}, \bibinfo{year}{2011}.
\newblock \bibinfo{title}{Uncovering the temporal dynamics of diffusion
  networks}, in: \bibinfo{booktitle}{Proceedings of the 28th International
  Conference on Machine Learning}, \bibinfo{address}{Washington, USA}. pp.
  \bibinfo{pages}{561--568}.
%Type = Inproceedings
\bibitem[{Gomez{-}Rodriguez et~al.(2010)Gomez{-}Rodriguez, Leskovec and
  Krause}]{Gomez-RodriguezLK10}
\bibinfo{author}{Gomez{-}Rodriguez, M.}, \bibinfo{author}{Leskovec, J.},
  \bibinfo{author}{Krause, A.}, \bibinfo{year}{2010}.
\newblock \bibinfo{title}{Inferring networks of diffusion and influence}, in:
  \bibinfo{booktitle}{Proceedings of the 16th {ACM} {SIGKDD} International
  Conference on Knowledge Discovery and Data Mining},
  \bibinfo{address}{Washington, DC, USA}. pp. \bibinfo{pages}{1019--1028}.
%Type = Article
\bibitem[{Goyal et~al.(2011)Goyal, Bonchi and Lakshmanan}]{GoyalBL11}
\bibinfo{author}{Goyal, A.}, \bibinfo{author}{Bonchi, F.},
  \bibinfo{author}{Lakshmanan, L.V.S.}, \bibinfo{year}{2011}.
\newblock \bibinfo{title}{A data-based approach to social influence
  maximization}.
\newblock \bibinfo{journal}{Proc. {VLDB} Endow.} \bibinfo{volume}{5},
  \bibinfo{pages}{73--84}.
%Type = Inproceedings
\bibitem[{He et~al.(2016)He, Xu, Kempe and Liu}]{HeX0L16}
\bibinfo{author}{He, X.}, \bibinfo{author}{Xu, K.}, \bibinfo{author}{Kempe,
  D.}, \bibinfo{author}{Liu, Y.}, \bibinfo{year}{2016}.
\newblock \bibinfo{title}{Learning influence functions from incomplete
  observations}, in: \bibinfo{booktitle}{Advances in Neural Information
  Processing Systems 29 ({NIPS} 2016)}, \bibinfo{address}{Barcelona, Spain}.
  pp. \bibinfo{pages}{2065--2073}.
%Type = Inproceedings
\bibitem[{Kempe et~al.(2003)Kempe, Kleinberg and Tardos}]{KempeKT03}
\bibinfo{author}{Kempe, D.}, \bibinfo{author}{Kleinberg, J.M.},
  \bibinfo{author}{Tardos, {\'{E}}.}, \bibinfo{year}{2003}.
\newblock \bibinfo{title}{Maximizing the spread of influence through a social
  network}, in: \bibinfo{booktitle}{Proceedings of the 9th {ACM} {SIGKDD}
  International Conference on Knowledge Discovery and Data Mining},
  \bibinfo{address}{Washington, DC, USA}. pp. \bibinfo{pages}{137--146}.
%Type = Book
\bibitem[{Mitzenmacher and Upfal(2005)}]{MitzenmacherU05}
\bibinfo{author}{Mitzenmacher, M.}, \bibinfo{author}{Upfal, E.},
  \bibinfo{year}{2005}.
\newblock \bibinfo{title}{Probability and Computing: Randomized Algorithms and
  Probabilistic Analysis}.
\newblock \bibinfo{publisher}{Cambridge University Press}.
%Type = Inproceedings
\bibitem[{Myers and Leskovec(2010)}]{MyersL10}
\bibinfo{author}{Myers, S.A.}, \bibinfo{author}{Leskovec, J.},
  \bibinfo{year}{2010}.
\newblock \bibinfo{title}{On the convexity of latent social network inference},
  in: \bibinfo{booktitle}{Advances in Neural Information Processing Systems 23
  ({NIPS} 2010)}, \bibinfo{address}{Vancouver, British Columbia, Canada}. pp.
  \bibinfo{pages}{1741--1749}.
%Type = Inproceedings
\bibitem[{Narasimhan et~al.(2015)Narasimhan, Parkes and
  Singer}]{NarasimhanPS15}
\bibinfo{author}{Narasimhan, H.}, \bibinfo{author}{Parkes, D.C.},
  \bibinfo{author}{Singer, Y.}, \bibinfo{year}{2015}.
\newblock \bibinfo{title}{Learnability of influence in networks}, in:
  \bibinfo{booktitle}{Advances in Neural Information Processing Systems 28
  ({NIPS} 2015)}, \bibinfo{address}{Montreal, Quebec, Canada}. pp.
  \bibinfo{pages}{3186--3194}.
%Type = Article
\bibitem[{Nemhauser et~al.(1978)Nemhauser, Wolsey and Fisher}]{NemhauserWF78}
\bibinfo{author}{Nemhauser, G.L.}, \bibinfo{author}{Wolsey, L.A.},
  \bibinfo{author}{Fisher, M.L.}, \bibinfo{year}{1978}.
\newblock \bibinfo{title}{An analysis of approximations for maximizing
  submodular set functions - {I}}.
\newblock \bibinfo{journal}{Math. Program.} \bibinfo{volume}{14},
  \bibinfo{pages}{265--294}.
%Type = Inproceedings
\bibitem[{Netrapalli and Sanghavi(2012)}]{NetrapalliS12}
\bibinfo{author}{Netrapalli, P.}, \bibinfo{author}{Sanghavi, S.},
  \bibinfo{year}{2012}.
\newblock \bibinfo{title}{Learning the graph of epidemic cascades}, in:
  \bibinfo{booktitle}{Proceedings of the 12th {ACM} {SIGMETRICS/PERFORMANCE}
  Joint International Conference on Measurement and Modeling of Computer
  Systems}, \bibinfo{address}{London, UK}. pp. \bibinfo{pages}{211--222}.
%Type = Inproceedings
\bibitem[{Pouget{-}Abadie and Horel(2015)}]{Pouget-AbadieH15}
\bibinfo{author}{Pouget{-}Abadie, J.}, \bibinfo{author}{Horel, T.},
  \bibinfo{year}{2015}.
\newblock \bibinfo{title}{Inferring graphs from cascades: {A} sparse recovery
  framework}, in: \bibinfo{booktitle}{Proceedings of the 32th International
  Conference on Machine Learning}, \bibinfo{address}{Lille, France}. pp.
  \bibinfo{pages}{977--986}.
%Type = Inproceedings
\bibitem[{Rosenfeld et~al.(2018)Rosenfeld, Balkanski, Globerson and
  Singer}]{RosenfeldBGS18}
\bibinfo{author}{Rosenfeld, N.}, \bibinfo{author}{Balkanski, E.},
  \bibinfo{author}{Globerson, A.}, \bibinfo{author}{Singer, Y.},
  \bibinfo{year}{2018}.
\newblock \bibinfo{title}{Learning to optimize combinatorial functions}, in:
  \bibinfo{booktitle}{Proceedings of the 35th International Conference on
  Machine Learning}, \bibinfo{address}{Stockholmsm{\"{a}}ssan, Stockholm,
  Sweden}. pp. \bibinfo{pages}{4371--4380}.
%Type = Article
\bibitem[{Shan et~al.(2019)Shan, Chen, Li, Sun and Zhang}]{ShanCLSZ19}
\bibinfo{author}{Shan, X.}, \bibinfo{author}{Chen, W.}, \bibinfo{author}{Li,
  Q.}, \bibinfo{author}{Sun, X.}, \bibinfo{author}{Zhang, J.},
  \bibinfo{year}{2019}.
\newblock \bibinfo{title}{Cumulative activation in social networks}.
\newblock \bibinfo{journal}{Sci. China Inf. Sci.} \bibinfo{volume}{62},
  \bibinfo{pages}{52103:1--52103:21}.

\end{thebibliography}





% \end{thebibliography}
\end{document}